%% file: main.tex
\documentclass[a4paper]{article}
\usepackage[a4paper, margin=2.54cm]{geometry}

\usepackage{physics}
\usepackage{amssymb,amsthm}
\usepackage[linesnumbered,lined,boxed,noend]{algorithm2e}
\usepackage{hyperref}
\usepackage{comment}
\usepackage{pgfplots}
\pgfplotsset{compat=1.9}
\usepackage{float}
\usepackage{lineno}

\makeatletter
\renewenvironment{algomathdisplay}
 {\[}
 {\@endalgocfline\vspace{-\baselineskip}\]\;}
\makeatother
\newcommand{\refAlgoStep}[1]{{\footnotesize\textsc{\textbf{\ref{#1}}}}}
\SetKwFunction{SBPEval}{SBPEval}
\SetKwFunction{ClassicalAdd}{ClassicalAdd}
\SetKwFunction{ClassicalProd}{ClassicalProd}
\SetKwFunction{Add}{Add}
\SetKwFunction{LinearCombination}{LinearCombination}
\SetKwFunction{MultiplyNoOffset}{MultiplyNoOffset}
\SetKwFunction{Multiply}{Multiply}
\DontPrintSemicolon
\SetKwInOut{Input}{input}\SetKwInOut{Output}{output}
\SetKwProg{Fn}{Function}{:}{end}

\theoremstyle{definition}
\newtheorem{defn}{Definition}[section]
\newtheorem{example}[defn]{Example}
\theoremstyle{plain}
\newtheorem{prop}[defn]{Proposition}
\newtheorem{coroll}[defn]{Corollary}
\theoremstyle{remark}
\newtheorem{remark}[defn]{Remark}

\usepgfplotslibrary{groupplots}
\usetikzlibrary{positioning}
\usetikzlibrary{external}
\tikzsetexternalprefix{compiled-figures/}

\usepackage{pgffor}
\usepackage{xcolor}

\DeclareMathOperator{\sgn}{sgn}
\DeclareMathOperator{\round}{rnd}
\DeclareMathOperator{\card}{card}

\floatstyle{boxed}
\newfloat{Formula box}{tbp}{lop}

\usepackage{setspace}
\date{}

\begin{document}

\title{Hybrid quantum floating-point method for sharp arithmetic}

\author{Gabriele Agliardi$^{*,1,2}$ and Enrico Prati$^{\dag,3,4}$\\[8mm]
{\footnotesize\parbox{.8\textwidth}{
$^*$ gabriele.agliardi@it.ibm.com $^\dag$ enrico.prati@unimi.it\\[2mm]
$^1$ IBM Quantum, IBM Research, Via Paolo Nanni Costa 30, I--40133 Bologna, Italy\\
$^2$ At the time of paper conception, also at Dipartimento di Fisica, Politecnico di Milano, Piazza Leonardo da Vinci I--20133 Milano, Italy\\
$^3$ Dipartimento di Fisica “Aldo Pontremoli”, Università degli Studi di Milano, Via Celoria 16, I–20133 Milano, Italy\\
$^4$ Istituto di Fotonica e Nanotecnologie, Consiglio Nazionale delle Ricerche, Piazza Leonardo da Vinci 32, I–20133 Milano, Italy\\
}}
}
\maketitle

\begin{abstract}
There are several possible ways to encode random variables in a quantum state.
The basis encoding of bit strings has paramount importance because it allows to load the values of a random variable through the superposition of corresponding basis states, and to then exploit quantum parallelism in processing algorithms. The basis encoding offers a natural way to represent an unsigned integer random variable, and extends to signed integers, as well as to fixed-point and floating-point variables.
Each quantum representation of fractional numbers, however, involves a trade-off between accuracy and depth of manipulation circuits. 
Here, an efficient hybrid quantum-classical representation of quantum floating points is introduced. 
It combines a quantum register containing the values, with a classical register storing global information about the variable, namely the range and approximation tolerances. 
The sum and product operations are defined, in such a way as to ensure they are performed without overflow. 
By taking advantage of the stored classical information, the precision degradation that occurs due to rounding after repeated data manipulations, can be significantly reduced compared to known strategies. 
Ad hoc examples show up to around $90\%$ reduction in approximation, compared to previous techniques, after repeated additions.
The method finds application in many algorithms of practical relevance and constitutes a significant advance in the design of arithmetic circuits with low depth and high accuracy.
\end{abstract}
\vspace{2pc}

\section{Introduction}
The precision of a floating point variable degrades significantly as a result of repeated arithmetic manipulations, when the register size is fixed and a no-overflow requirement is guaranteed on all data points.
Here we investigate error propagation on quantum floating point numbers when applying summation and multiplication. We propose a representation for fractional numbers, which can reduce such errors by keeping track of global properties through classical registers.
The relevance of floating-point quantum arithmetic is fundamental in many application fields.
Indeed, by mimicking classical arithmetic \cite{rieffel_quantum_2011, nielsen_quantum_2010}, quantum floating-point calculations can evaluate virtually any function \cite{haner_optimizing_2018}, thus allowing quantum computers to tackle data transformations of arbitrary complexity, as long as the depth and width of the circuit are supported by the hardware.
In finance, for example, quantum floating point numbers are needed to represent complex payoffs in option pricing~\cite{doriguello_quantum_2022}, as opposed to single-time and piecewise linear payoffs that can be approximated by native quantum amplitude operations \cite{woerner_quantum_2019,maronese2024quantum}.

Previously, we studied the efficient loading of multivariate distributions by quantum generative adversarial networks \cite{agliardi2022optimal, agliardi2022optimized}. There, the technique finds application in the data preparation that occurs prior to the data manipulation discussed in the present manuscript. Next, we have isolated the data encoding process as a separate abstraction layer, to improve the design of quantum circuits~\cite{agliardi2025quantum}. 
In turn, data processing enabled by hybrid quantum floating-point arithmetic can be used to compute functions for quantum integration via quantum amplitude estimation~\cite{agliardi2022quantum}.

Most efforts in the quantum arithmetic community are concerned with designing efficient circuits to perform quantum modular integer addition and multiplication \cite{vedral_quantum_1996, draper_addition_2000, cuccaro_new_2004, thapliyal_design_2013, nguyen_resource-efficient_2014, ruiz-perez_quantum_2017, haner_optimizing_2018, seidel_efficient_2022}, which also provide the building blocks of real number arithmetic \cite{haener_quantum_2018, seidel_efficient_2022}. Borrowing from the classical case, real numbers can be represented in the quantum context by either fixed-point \cite{yang_quantum_2022} or floating-point \cite{haener_quantum_2018}. 

Recently, Seidel et al.~\cite{seidel_efficient_2022} introduced a variant of the floating-point representation, called \textit{mono-quantum coding}, in which only the mantissa is stored in a quantum register while the exponent is classical, and they claim that this strategy allows for significant savings in depth and width. Mono-quantum coding is specifically suitable if all data values stored in superposition in the register share a similar scale, so that applying the same exponent to all data points does not imply a significant approximation.

This assumption often holds in practice, since a given register is typically dedicated to homogeneous quantities. Mono-quantum coding can be seen as a first example of a classically enriched register, since some global properties of the data set (here the exponent), are stored classically to improve performance under certain conditions - in this case their similar scale.

Here we further develop the idea of using classical side registers, showing how other global information can be stored classically to reduce error propagation when multiple arithmetic operations are performed.
We employ our representation to store a random variable, where the random outcomes are represented as basis states, and their probabilities as squared amplitudes. In this setting, one can load a target random variable into a quantum computer \textit{approximately}, in consequence of the limited number of available qubits, and therefore the hybrid register can track the maximal gap -- henceforth called tolerance, between the target random variable and its representation by the quantum state.
Our hybrid representation stores the global scaling factor -- which is applied to all random outcomes, the global offset, and the tolerances, respectively, in a classical register. The hybrid variables are then equipped with a sum and product operation, which are ensured to be performed without overflow, by appropriately adjusting the scale of the output in accordance with the number of available qubits in the output register. 
Furthermore, when performing additions or multiplications, tolerance propagation is taken into account. A tolerance region is therefore provided, in which the target sum or product lies, around the approximate result represented by the quantum state.
The encoding is benchmarked against previous techniques, among which Ref.~\cite{seidel_efficient_2022} provides the most comparable technique. With ad-hoc examples, we show
a reduction in the size of tolerance windows up to about $90\%$ obtained after repeated applications of addition.

The Results section introduces the concept of \textit{quantum variables}, that provides a useful abstraction for data representations based on basis encodings. Using the formalism of quantum variables, we then propose our hybrid encoding, called \textit{classically-enriched floating point variables} (CEFV). We show how to use CEFVs to represent random variables. We consistently define both addition and multiplication, and we discuss the main parameters that influence them. Next, we evaluate the advantages of the proposed method, compared to the state of the art.
The Discussion section contains a summary of the implications and an outlook on future work.
The Methods section provides insight into how the addition and multiplication algorithms work.
The formal definition of quantum variables, the detailed algorithms, the correctness proofs, as well as some examples and extended results, are collected in the Supplementary Information~(SI) materials.

\section{Results}

\subsection{Definition of quantum variable and classically-enriched floating point variable}
In classical computer science, a \textit{variable} is a memory cell equipped with a \textit{type}, such as unsigned or signed integer, or floating point. The type creates a mapping between binary strings contained in the cell and values in a target domain. In the case of a quantum register, we can similarly define a \textit{quantum variable} as a quantum register of $n$ qubits, equipped with an encoding function $g$ which induces a mapping fomr the basis states, labelled as $\ket{0}, ..., \ket{2^n-1}$, to the values in a given \textit{domain} $\mathcal D$. For instance, when $g$ is the identity, the commonly used unsigned integer quantum variables are obtained. Similarly, when $g$ is the complement function of the two, signed integers are produced. One can also easily introduce fixed-point or floating-point numbers.

In the following, we address how to benefit from mappings generated by encoding functions $g$ that depend on some classical parameters. In particular, such parameters can be modified throughout the computation, so that operations on quantum variables can be obtained by altering the classical parameters, by applying quantum circuits to the registers, or both. Specifically, we introduce a new quantum variable type called Classically-Enriched Floating Point Variable (CEFV), which is initially defined by an affine transformation with an offset $a \in \mathbb R$ and a scaling factor $b \in \mathbb R \setminus \{0\}$. The encoding function $g$ maps the basis state $\ket{z}$ to the value $x = a+bz \in \mathcal{D}_{a,b}^{\mathrm{fp}} = \{ a+bz \; | \;  z =0, ..., 2^n-1\}$, see Fig.~\ref{fig:cefv}. The definition of CEFV is completed with two additional tolerance parameters $\epsilon^\pm$ discussed in the next subsection.

Being defined on basis states, the quantum variables induced by whatever mapping based on an encoding function $g$, share the powerful ability to perform operations in \textit{quantum parallel}. Namely, if an algorithm is able to perform (say) the sum of two quantum variables on basis states, then the operator extends by linearity on combinations of basis states, and computes the sum of all the basis states involved in such combination \textit{in a single application}. 
The possibility of quantum parallelism is a fundamental difference compared to classical computation. In fact, quantum parallelism underlies many algorithm designs~\cite{nielsen_quantum_2010}.

\begin{figure}
\small\centering
\fbox{
\includegraphics{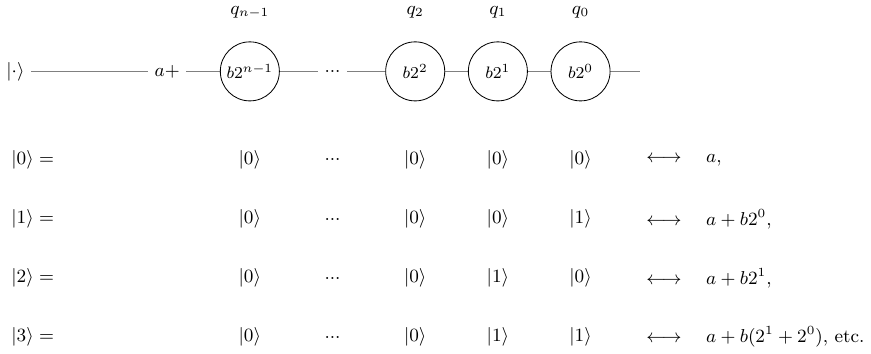}
}
\caption{The top row contains a visual representation of a CEFV. In the rows below, some examples of the floating-point number represented by the CEFV, when the quantum state belongs to the computational basis. States in the computational basis follow the usual bit string decomposition, thus providing a representation of any $x \in \mathcal{D}_{a,b}^{\mathrm{fp}}$.}\label{fig:cefv}
\end{figure}

\subsection{Encoding of a random variable by a quantum variable}
Given a quantum variable $\mathcal F$ and letting $\{\ket{z}\}_{z=0}^{2^n-1}$ be the basis states of its register, it is natural to associate each of its states $\ket{\psi}$ with the discrete random variable $X_{\ket{\psi}}$ valued in $\{g(z)\}_{z=0}^{2^n-1}$, with the probabilities $p_{g(z)} = \abs{\braket{\psi}{z}}^2$, as the probability of measuring $z$ in the quantum register coincides with $p_{g(z)}$.

Vice versa, a random variable $X$ can be encoded in a quantum state of a quantum variable. Then, it is possible to manipulate the entire random variable in quantum parallel. If the domain of the random variable is compatible with that of the quantum variable, the encoding is straightforward: $\ket{\psi_X} = \sum_{x \in \mathcal D} \sqrt{f_X(x)} \ket{g(x)}$, where $f_X(x)$ is the probability of $X=x$. When the domain of the random variable is not the same, and is potentially continuous, it is possible to encode the random variable approximately.
Consequently, we equip CEFVs with two additional tolerance parameters $\epsilon^\pm$. A random variable $Y$ is said \textit{$\mathcal{F}$-compatible with the state $\ket{\psi}$} if $- \epsilon^{-} \leq Y-X_{\ket{\psi}} \leq \epsilon^+$.

Similarly, given two possibly dependent random variables $X_1$ and $X_2$, compatible with the domain of two quantum variables $\mathcal F_1$ and $\mathcal F_2$ respectively, one can represent them as a state based on the two quantum registers, $\ket{\psi_{(X_1, X_2)}}_{1,2} = \sum_{x_1 \in \mathcal D_1, x_2 \in \mathcal D_2} \sqrt{f_{X_1, X_2}(x_1, x_2)} \ket{g_1 (x_1)} \ket{g_2 (x_2)}$, where $f_{X_1, X_2}$ is the joint probability. 
Note that the dependence between the random variables is equivalent to the state being entangled across the two registers. One can introduce tolerances also for the multi-variable case, more specifically by defining that a random variable $Y_1$ is $\mathcal{F}_1$-compatible with a state $\ket{\psi}_{1,2}$, living in a superset of the qubits of the register of $\mathcal F_1$.
All the concepts above are more formally treated in Section~S1 of the~SI. Fig.~S1 visually represents how a random variable is approximated by a quantum variable.

\subsection{Arithmetic operations on CEFVs and consistency with approximate random variables}
The usefulness of the framework above is that it results consistent when implementing both the addition and the multiplication operations. The addition algorithm for CEFVs takes two CEFVs $\mathcal{F}_1$ and $\mathcal{F}_2$\footnote{%
    Here and throughout the work, notations $a_\star, b_\star,$ etc. are implicitly referred to $\mathcal{F}_\star$. For instance, $b_1$ is the factor of $\mathcal{F}_1$. %
} and one scaling factor $b_\mathrm{lead}$ as inputs (see Methods).
It outputs an offset $a_\mathrm{out}$, a scaling factor $b_\mathrm{out}$, and the tolerances $\epsilon_\mathrm{out}^\pm$, thus defining the CEFV $\mathcal F_\mathrm{out}$. It also outputs a quantum circuit $U_\mathrm{add}$, operating on the registers of $\mathcal F_1$, $\mathcal F_2$, and $\mathcal F_\mathrm{out}$, as well as on an auxiliary register $\ket{\cdot}_\mathrm{aux}$.
For any input state $\ket{\psi_\mathrm{in}}_{1,2}$ defined over the registers of $\mathcal F_1$ and $\mathcal F_2$, if $\ket{\psi_\mathrm{in}}_{1,2}$ is $\mathcal F_i$-compatible with two random variables $Y_i$ ($i=1,2$), then $U_\mathrm{add} \ket{\psi_\mathrm{in}}_{1,2} \ket{0}_\mathrm{out} \ket{0}_\mathrm{aux} = \ket{\psi_\mathrm{out}}_{1,2, \mathrm{out}} \ket{0}_\mathrm{aux}$, and the output state $\ket{\psi_\mathrm{out}}_{1,2, \mathrm{out}}$ keeps being $\mathcal F_i$-compatible with two random variables $Y_i$ ($i=1,2$), and is also $\mathcal{F}_\mathrm{out}$-compatible with $Y_\mathrm{out} = Y_1+Y_2$.
Moreover, the output $b_\mathrm{out}$ is $b_\mathrm{lead}$, modulo some `qubit shifting' that allows the optimal exploitation of all the qubits in the output register. More precisely, $b_\mathrm{out} / b_\mathrm{lead}$ is a power of 2, appropriately chosen to minimize the tolerances $\epsilon^\pm_\mathrm{out}$ in consideration of the available qubits in the output register, while ensuring no overflow.

Similarly, the multiplication algorithm with the same properties is given: it takes two input CEFVs and a reference scaling factor, and provides an output CEFV and a quantum operator performing the product. Again, the output scaling factor differs from the reference scaling factor by a power of 2, and it is chosen to best exploit the available qubits, without overflow. The similarity between addition and multiplication is due to the fact that both algorithms are grounded on the evaluation of semi-boolean polynomials, as discussed in the Methods.

\subsection{Choice of the output scaling factor in the sum}

\begin{figure}
    \input{Figures/blead-summary-small.tex}
\caption{Effect of the choice of $b_\mathrm{lead}$ on errors, for the sum of two CEFVs of size $n$, in an output register of same size $n=8$. In the top plot, $b_1$ is kept equal to 1, $b_2$ varies from $10^{-3}$ to $10^4$, and $b_\mathrm{lead}$ in $[1,2]$. %
The color intensity represents the relative error, namely $\frac{\epsilon^++\epsilon^-}{b_1 2^n + b_2 2^n}$. In the bottom plot, a focus on three specific choices of $b_\mathrm{lead}$, namely $b_1$ (i.e., 1), $b_2$, and the optimal $b_\mathrm{opt}$. %
} \label{fig:b-lead-summary}
\end{figure}
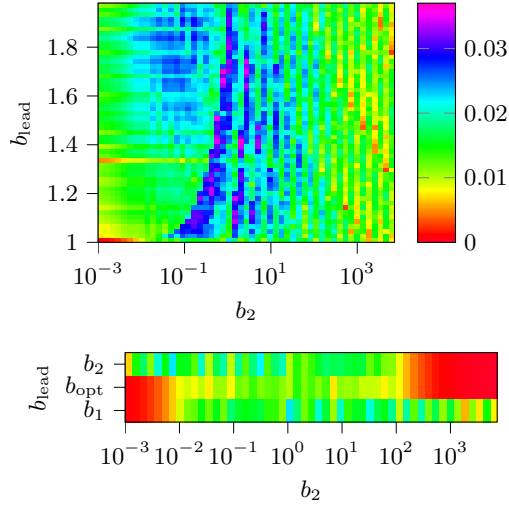
Focusing on addition for simplicity, we now study how the choice of $b_\mathrm{lead}$ affects the sum through a numerical campaign. We specifically address the error, relative to the scale of the inputs, as a function of $b_1$, $b_2$ and $b_\mathrm{lead}$. For simplicity, we consider the case $n_1 = n_2 =: n$, and hence we define the relative error as $\frac{\epsilon_{\mathrm{approx}}^++\epsilon_{\mathrm{approx}}^-}{b_1 2^n + b_2 2^n}$. Observe that we can restrict ourselves to $b_1=1$, since the relative error is invariant for transformations $(b_1, b_2, b_\mathrm{lead}) \mapsto (\alpha b_1, \alpha b_2, \alpha b_\mathrm{lead})$ for any $\alpha \in \mathbb R \setminus \{0\}$. Additionally, $b_\mathrm{lead}$ can be restricted to the range $[1,2]$, since $2 b_\mathrm{lead}$ provides the same outcome as $b_\mathrm{lead}$.
Results are visualized in the heat maps at the top of Fig.~\ref{fig:b-lead-summary}.

Next, define $b_\mathrm{opt}$ as the optimal value of $b_\mathrm{lead}$ for a given $b_2$. Numerically, $b_\mathrm{lead}$ is obtained by a global search algorithm. Our choice falls on the dual annealing algorithm implemented in Scipy~\cite{2020SciPy-NMeth}. In the bottom heat maps of Fig.~\ref{fig:b-lead-summary}, we compare the choice $b_\mathrm{lead} = b_1 = 1$ and $b_\mathrm{lead} = b_2$ against the optimal $b_\mathrm{lead} = b_\mathrm{opt}$. As expected, the error is lower when one input dominates the other. The best between $b_1$ and $b_2$ is anyway a good choice even for intermediate cases, making the in-place variant introduced in the Methods a viable option. %

The effect of $b_\mathrm{lead}$ is further discussed in the SI in Sec~S3.1. Specifically, Figs.~S4 and~S5 exemplify the effect of $b_\mathrm{lead}$ on the circuit depth, in different settings of $b_2$. Fig.~S6 offers different visualizations for the information in Fig.~2, and extends it to different values of $n$.

\subsection{Advantages of introducing the offset}
Figure~\ref{fig:step} contains a visualization of the repeated application of the sum of a random variable with itself, when the variable is represented with and without offsets. The vertical dashed grids correspond to numbers that can be stored exactly in quantum states. The introduction of offsets allows for a denser grid in a specific area, implying a more precise representation of random variables in the same number of qubits. At the same time, offsets also limit the error propagation when applying arithmetic operations, as evidenced by the much smaller tolerance region in the figure. In the engineered case therein represented, the maximal error is decreased from $\epsilon^+ = 28$ to $\epsilon^+ = 3$, while $\epsilon^-=0$, namely a reduction of $89\%$.
\begin{figure}
    \small
    \centering
    \input{Figures/overflow-step-smallpic.tex}
    \caption{Cumulative distribution function for repeated applications of the sum. Left plot: the initial random variable (providing $6$ or $7$, both with $50\%$ probability), which can be encoded exactly in $3$ qubits both with or without offset. Right plot: after 6 applications of the sum of the initial random variable with itself, the target random variable (in dashed green) compared to the approximating random variable (namely the one encoded in the quantum states), as well as the approximation region (defined as the approximating random variable +/- tolerances). Vertical grids show the values which can be perfectly represented by the quantum variables. The advantage of offsets is clearly visible, in terms of a narrower approximation region.}
    \label{fig:step}
\end{figure}
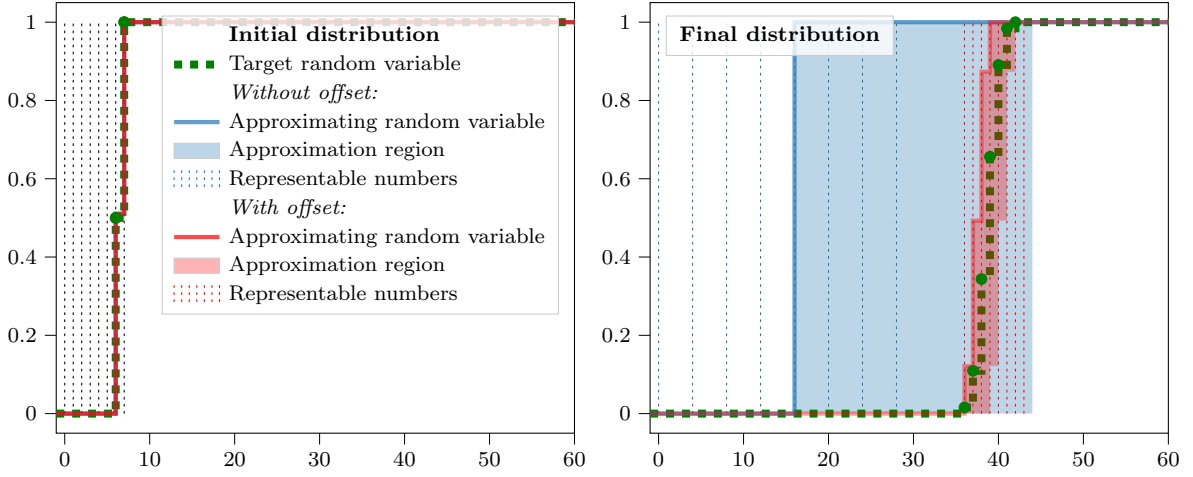
Fig.~S7 in the SI extends Fig.~3, showing the effect of \textit{each} application of the sum, with and without offset.

\subsection{Benchmarking CEFVs against mono-quantum floating variables}\label{subsec:benchmarking}
Let us finally compare the proposed approach against the \textit{mono-quantum floating variable} (MFV)~\cite{seidel_efficient_2022} on $n$ qubits. By this term, we refer to the quantum variable $(\ket{\cdot}, \mathcal{D}_{d}^{\mathrm{mono}})$ with quantum mantissa and classical exponent $d \in \mathbb{Z}$, with a qubit dedicated to the sign. The domain is the uniformly spaced grid
$$\mathcal{D}_{d}^{\mathrm{mono}} = \left\{ 2^d x' \; | \; x' \in \mathcal{D}^{\mathrm{int}}\right\} = \left\{ - 2^d 2^{n-1}, ..., - 2^d, 0, ..., 2^d 2^{n-1} -1 \right\},$$
where $\mathcal{D}^{\mathrm{int}} := \{ -2^{n-1}, ..., -1, 0, ..., 2^{n-1}-1 \}$ is the domain for signed integers in $n$ qubits, while the encoding is the two's complement of $x'$:
$$g_{d}^{\mathrm{mono}}(x) := g^\mathrm{tc} \left( 2^{-d} x \right) = g^\mathrm{tc} \left( x' \right),
\qquad
g^\mathrm{tc}(x') :=
\begin{cases}
    x' & \text{if } x' \geq 0,\\
    2^n+x' & \text{otherwise}.
\end{cases}
$$

The generality is motivated by the fact that an $n$-qubit MFV with exponent $d$ encodes the same variables as a $(n-1)$-qubit CEFV with $a=-2^d 2^{n-1}$ and $b=2^d$. In contrast, given a random variable encoded by an $n$-qubit CEFV, it is not possible in general to encode it in a $(n+1)$-qubit MFV. In fact, it may be impossible to encode the same random variable exactly in an MFV even with an unlimited number of qubits. The ability of CEFVs to represent random variables in a specific range far from zero with higher precision by few qubits, makes such variables suitable for many quantum applications, in a context where the circuit width and depth are limited by hardware capabilities.

Moreover, summing a classical number to the values of a CEFV, as well as multiplying the values of a CEFV by a classical constant, are purely classical operations. Converesely, the same operations require the application of quantum circuits in the context of MFVs. So, the proposed method eliminates the consumption of critical resources for these manipulations.

The addition or multiplication of two $n$-qubit variables may appear more expensive for CEFVs than they are for MFVs, since the conversion from $b_1, b_2$ to $b_\mathrm{lead}$ is costly if $b_1/b_2$, $b_1/b_\mathrm{lead}$ and $b_2/b_\mathrm{lead}$ are not multiples of 2. Two considerations should be made on this topic. First, the additional overhead may be balanced by the reduced number of qubits required to obtain a similar precision. Secondly, the overhead can be reduced or canceled by choosing $b_1$, $b_2$ and $b_\mathrm{lead}$ appropriately, where the best choice depends on the processing that the variables undergo in the entire quantum circuit, also beyond addition and multiplication.

In summary, CEFVs provide two more degrees of freedom compared to MFVs. On one side, the introduction of the offset, and on the other, the removal of the constraint that factors should be powers of 2. The effect is an increased precision of the representation, under fixed register size. The former change alone (namely, the introduction of the offset under the assumption that scaling factors are powers of 2) implies no overhead for the addition, compared to MFV, since the summation of offsets is accounted classically. On the other hand, we discuss in the Method section that the offsets do introduce a modest overhead on the product. The latter change instead, i.e. the possibility of arbitrary scaling factors, may have a more detrimental impact on performance, and consequently the choice of the scaling factors should be planned with care when designing the overall circuit, which typically contains other processing steps on top of the addition and multiplication described here.

\section{Discussion}\label{sec:discussion}
The above quantum-based representation of random variables valued in real domains differs from known encoding methods for the presence of a classical offset and a classical scaling factor. The definition of such representation also includes a below- and an above-tolerance, to keep track of the distance between the encoded data and the target random variable. Said representation is consistently equipped with a sum and a product operation. The operations are ensured to be performed without overflow, and they automatically tune the scaling factor of the output register in order to achieve an optimal usage of the available qubits, thus minimizing the tolerances associated to the output.
The method proves successful against known techniques of hybrid quantum-classical floating-point processing, with a reduction in the size of tolerance windows up to $89\%$ according to the test cases provided.
The method ensures that all states in the computational basis agree with the correct sum and product, up to a specified tolerance range.

\section{Methods}
The methods to embody both the addition and the multiplication are explained below. For a formal treatise, including proofs and examples, we refer to Section~S2 of the SI.

\subsection{Addition}
The summation algorithm takes as input two CEFVs $\mathcal{F}_1$ and $\mathcal{F}_2$, the size of the output register $n_\mathrm{max}$, and the desired scaling factor for the output, $b_\mathrm{lead}$. Notice that, if $b_\mathrm{lead}$ is very different from $b_1$ and $b_2$, the result may cause overflow or strong underflow. As a consequence, $b_\mathrm{lead}$ is adjusted by the algorithm to a $b_\mathrm{out}$ such that $b_\mathrm{lead}/b_\mathrm{out}$ is an appropriate power of 2 with signed integer exponent, such that overflow is ensured not to happen, and the available qubits in the output registers are maximally exploited to minimize rounding.

To add two CEFVs $\mathcal{F}_1$ and $\mathcal{F}_2$, the offsets can be simply summed together, while states must be processed by a quantum circuit that takes into account the scaling factors.
More precisely, the rescaling factor for $\mathcal{F}_i$ ($i=1,2$) is $b_\mathrm{out} / b_i$, and needs to be rounded to a binary fraction, to fit the representation possibilities of the target register. Said binary fraction is classically known, so the sum of the two inputs gets translated into a weighted sum of qubit values, where weights are integers. Therefore, one relies on the evaluation of polynomials in boolean variables with integer coefficients, a method introduced in Ref.~\cite{seidel_efficient_2022} with the name of \textit{semi-boolean polynomial evaluation}, here called \SBPEval for conciseness. More details are provided in Sec.~S2.1 of the SI. Specifically, the protocol is formalized as Algorithm~S1 and exemplified in Figs.~S2, S3 and S8.

The \Add algorithm produces an output CEFV $\mathcal F_\mathrm{out}$, by defining $n_\mathrm{out}$, $a_\mathrm{out}$ and $b_\mathrm{out}$, as well as the error thresholds $\epsilon_\mathrm{out}^\pm$, which reflect the rounding applied to the weights. The new CEFV $\mathcal F_\mathrm{out}$ has maximal size $n_\mathrm{max}$ qubits, and is such that $b_\mathrm{out} / b_\mathrm{lead}$ is a power of 2, with signed integer exponent. It also provides a quantum circuit $U$ that is able to calculate the sum of the inputs in the following sense: $U$ applied to an input state $\ket{\psi_\mathrm{in}}_{1,2}$ (living in the registers $\ket{\cdot}_1$ and $\ket{\cdot}_2$), calculates the sum of the two underlying random variables, up to a tolerance that is compatible with $\epsilon_\mathrm{out}^\pm$. In formulas, given $Y_i \simeq_{\mathcal{F}_i} \ket{\psi_\mathrm{in}}$ for $i=1,2$, the operator $U$ acts on such state, plus an output register which is assumed to be initialized to $\ket0$ and an auxiliary register that is also supposed to start in $\ket0$, and
\begin{equation}\label{eq:sum-correctness}
    U \ket{\psi_\mathrm{in}}_{1,2} \ket{0}_\mathrm{out} \ket{0}_\mathrm{aux} = \ket{\psi_\mathrm{out}}_{1,2,\mathrm{out}} \ket{0}_\mathrm{aux} \quad \text{such that} \quad
    \begin{cases}
        Y_i \simeq_{\mathcal{F}_i} \ket{\psi_\mathrm{out}} \text{ for } i=1,2,\\
        Y_1 + Y_2 \simeq_{\mathcal{F}_\mathrm{out}} \ket{\psi_\mathrm{out}}.
    \end{cases}
\end{equation}

Notice that $U$ depends on $\mathcal F_1$ and $\mathcal F_2$, namely on $a_i, b_i, \epsilon_i^\pm$ for $i=1,2$, but not on the variables $Y_1$ and $Y_2$, nor on the states $\ket{\psi_1}$ and $\ket{\psi_2}$, which are unknown when the circuit is generated. The validity of Eq.~\eqref{eq:sum-correctness} is the correctness of the algorithm, which is proven in Prop.~S2.9 of the SI and exemplified in Fig.~S3.

\subsection{In-place addition and other extensions}
A variant \cite{seidel_efficient_2022} of \SBPEval can calculate the sum of two unsigned integer registers in-place in the (say) first register, namely there exists an algorithm generating a unitary $U$ such that
$$U \ket{z_1}\ket{z_2} \ket{0} = \ket{z_1+z_2}\ket{z_2} \ket{0}$$
under the assumption of no overflow, for all basis states $\ket{z_1}$ and $\ket{z_2}$. This variant implies a small overhead in the multiplicative constants of the gate count, that does not affect the asymptotic behavior when the register sizes grow. Such variant easily extends to a weighted sum, as long as weights apply only to qubits in the second register, which is our case once we select $b_\mathrm{lead}=b_1$. In other words, we can in turn define a variant of our algorithm that operates in-place, once we guarantee that the target register $\ket{\cdot}_1$ is large enough to contain the result of the sum. %
As a consequence, a good strategy to minimize the number of qubits needed, is picking the \textit{leading register}, complementing it with additional qubits to ensure no overflow, and make in-place addition. By leading register we mean the input register that can contain the highest value, namely the register that maximizes $\abs{b} 2^{n}$, or, as a subordinate criterion, the register with more qubits.

Given our definition of the CEFV, one can trivially provide the appropriate \ClassicalAdd and \ClassicalProd operations, that take a CEFV and a real number in input, and output a new CEFV operating on the same register, with updated characteristics $a$ and $b$, and with no need for manipulating the quantum states. Combining \ClassicalProd with \Add, one can calculate arbitrary linear combinations \LinearCombination of two CEFVs, and specifically the difference.

The \Add and \LinearCombination algorithms easily extend to multiple input CEFVs. The multi-variable algorithm has in general smaller approximation error, compared to the repeated application of pairwise sums.

\subsection{Multiplication of two CEFVs}
Given two input random variables $X_1 = a_1 + b_1 Z_1$ and $X_2 = a_2 + b_2 Z_2$, their product is $X_1 X_2 = (a_1 a_2) + (a_1 b_2 Z_2 + a_2 b_1 Z_1 + b_1 b_2 Z_1 Z_2)$, so a way to compute the integer product $Z_1 Z_2$ is needed, and then the problem reduces to the calculation of a \LinearCombination of CEFVs. On the other hand, also the product $Z_1 Z_2$ can be implemented as an \SBPEval~\cite{seidel_efficient_2022}. This idea can be further improved by defining a single semi-boolean polynomial whose evaluation directly provides the desired result, thus giving rise to our \Multiply method. While addition uses a semi-boolean polynomial that is linear in its variables, the product requires a quadratic polynomial, due to the presence of the term $Z_1 Z_2$. This fact is manifest in the implementation formalized in Algorithm~S2 and discussed more broadly in Sec.~S2.3 of the SI. In the case of addition, the introduction of offsets implies no overhead, as the sum of the offsets can be accounted classically. In the case of the multiplication, instead, offsets cause the addition of the linear terms $a_1 b_2 Z_2$ and $a_2 b_1 Z_1$, and thence an increase in the number of quantum gates in the circuits. However, offsets have no impact on the most complex terms, namely the quadratic form $b_1 b_2 Z_1 Z_2$, thus limiting their negative effect on the circuit depth.

\section*{Acknowledgements}
E.P. acknowledges IBM for having supported this work with IBM Quantum access.

\section*{Funding}
This research received no external funding.

\section*{Data availability}
The data that support the findings of this study are available from the corresponding author upon reasonable request.

\section*{Code availability}
The code and the algorithm used in this study are available from the corresponding author upon reasonable request.

\section*{Author Contributions}
Both G.A. and E.P. contributed to elaborate and discuss the results and to the writing of the manuscript.

\section*{Competing Interests statement}
The Authors declare that the patent application P202300674US01 "Quantum variables implementation" is pending.

\clearpage

\appendix

\input{main-si}

\end{document}

%% file: Figures/blead-summary-small.tex
\centering
\begin{minipage}{.46\textwidth}
\vspace{1mm}
\centering
\includegraphics{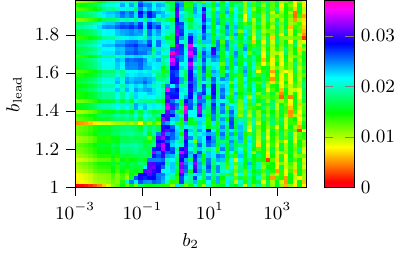}
\\[2mm]
\includegraphics{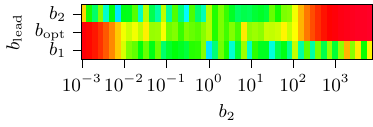}
\end{minipage}

%% file: Figures/overflow-step-smallpic.tex
\includegraphics{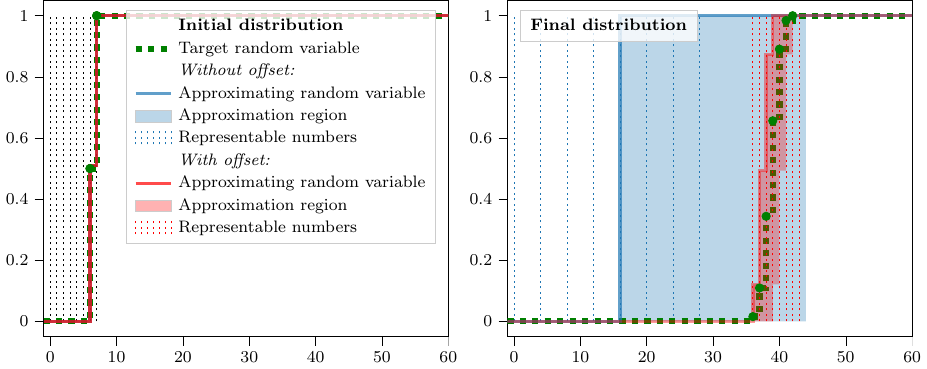}

%% file: main-si.tex
\makeatletter
\renewenvironment{algomathdisplay}
 {\[}
 {\@endalgocfline\vspace{-\baselineskip}\]\;}
\makeatother
\SetKwFunction{SBPEval}{SBPEval}
\SetKwFunction{ClassicalAdd}{ClassicalAdd}
\SetKwFunction{ClassicalProd}{ClassicalProd}
\SetKwFunction{Add}{Add}
\SetKwFunction{LinearCombination}{LinearCombination}
\SetKwFunction{MultiplyNoOffset}{MultiplyNoOffset}
\SetKwFunction{Multiply}{Multiply}
\DontPrintSemicolon
\SetKwInOut{Input}{input}\SetKwInOut{Output}{output}
\SetKwProg{Fn}{Function}{:}{end}

\theoremstyle{definition}
\theoremstyle{plain}
\theoremstyle{remark}

\usepgfplotslibrary{groupplots}
\usetikzlibrary{positioning}
\usetikzlibrary{external}
\tikzexternalize %
\tikzsetexternalprefix{compiled-figures/}

\floatstyle{boxed}
\newfloat{formulabox}{tbp}{lop}
\floatname{formulabox}{Formula box}

\renewcommand\thefigure{S\arabic{figure}}
\renewcommand\thesection{S\arabic{section}}
\renewcommand\thealgocf{S\arabic{algocf}}
\renewcommand\theformulabox{S\arabic{formulabox}}

\section{Quantum variables}
\subsection{Encoding of a random variable and concept of quantum variable}
In quantum computing, given a basis state $\ket{z_0} \cdots \ket{z_{n-1}}$ on $n$ qubits, where $z_i$ are binary digits, it is a common practice to represent it as $\ket{z}$, where $z$ is the unsigned integer $z = \sum_j z_j 2^j$ in $\{ 0, ..., 2^{n}-1\}$. In turn, $\ket{z}$ can be interpreted to be a quantum representation of another number $x$ in a domain $\mathcal{D}$, for instance $\{ 2^{n-1}, ..., 0, ..., 2^{n-1}-1\}$, through a bijective mapping $g$ between $x$ and $z$. Once $g$ is given, it induces a bijection between the random variables valued in $X$ and the quantum states on $n$ qubits, as we formalize in the following Definition.

\begin{defn}[$g$-encoding of a random variable]\label{SIdefn:g-encoding}
Let $n$ be a positive integer, and $\mathcal{D}$ be a discrete real domain with $2^n$ points. Let $g$ be a bijection between $\mathcal{D}$ and $\{0, ..., 2^n-1\}$. Moreover, let $X$ be a random variable valued in $\mathcal{D}$, with a mass probability function $f_X (x) = \mathbb{P}(X=x)$. Then we say that an $n$-qubit state $\ket{\psi_X}$ is the \textit{$g$-encoding for $X$} if
$$\ket{\psi_X} = \sum_{x \in \mathcal{D}} \sqrt{f_X(x)} \ket{g(x)}.$$
\end{defn}

Vice versa, given a state one can define its corresponding random variable:
\begin{defn}[random variable $g$-encoded by a state]\label{SIdefn:g-encoded}
Let $n$ be a positive integer, and $\mathcal{D}$ be a set of $2^n$ real numbers. Let $g$ be a bijection between $\mathcal{D}$ and $\{0, ..., 2^n-1\}$. Finally, let $\ket{\psi}$ be a state on $n$ qubits. The random variable $X_{\ket{\psi}}$ having density
$$f_{X_{\ket{\psi}}}(x) = \abs{\braket{\psi}{g(x)}}^2 \quad \forall x \in \mathcal{D}.$$
is called the \textit{random variable $g$-encoded by $\ket{\psi}$}.
\end{defn}

Notice that, given $\mathcal{D}$, $g$ and perfect knowledge of a state $\ket{\psi}$, one can obtain the random variable $g$-encoded by $\ket{\psi}$, called $X_{\ket{\psi}}$, and then produce its $g$-encoding $\ket{\psi_{X_{\ket{\psi}}}}$. If the original state $\ket{\psi}$ has non negative amplitudes, then $\ket{\psi} = \ket{\psi_{X_{\ket{\psi}}}}$.

Vice versa, starting from a random variable $X$, one can build its $g$-encoding $\ket{\psi_X}$ and then the respective $g$-encoded random variable $X_{\ket{\psi_X}}$. Then, $X=X_{\ket{\psi_X}}$.

Definitions~\ref{SIdefn:g-encoding} and~\ref{SIdefn:g-encoded} are later generalized to the context of multiple dependent random variables, see Definitions~\ref{SIdefn:g-encoding-mult} and~\ref{SIdefn:g-encoded-mult} respectively. Before, though, let us formalize the concept of a quantum variable.

It is practically handy to associate the encoding function $g$ with the quantum register holding the information, thus giving rise to the following definition.
\begin{defn}[quantum variable]
A \textit{quantum variable of size $n$} is a tuple $(\ket{\cdot}, \mathcal D, g)$, where $\ket{\cdot}$ is a quantum register of size $n$, $\mathcal{D}$ is a set of $2^n$ real values called \textit{domain}, and $g$ is a bijection between $\mathcal{D}$ and $\{ 0, ..., 2^n-1\}$ called \textit{encoding function}.
\end{defn}

 We use the expression quantum \textit{variable} in analogy with the term used in classical programming languages, and similar to Refs.~\cite{silq, qrisp}. It should be noted though that a quantum variable does not correspond to a random variable. Rather, a random variable is the \textit{value} of a quantum variable. In other words, a quantum variable is a structure that can contain one random variable at a time, as much as a classical variable is a cell containing one number at a time.

\begin{example}[unsigned integer quantum variables]
As an initial example, consider a register of $n$ qubits, and consider $g$ as the identity map on $\{0, ..., 2^n-1\}$. In this case, basis states are used to encode unsigned integer numbers, and real-amplitude states encode random variables values in the unsigned integer domain. Thus, \textit{unsigned integer quantum variables} are defined.
\end{example}

Such an encoding of unsigned integers is widespread in quantum computing, as it allows to perform quantum parallel calculations. In other words, suppose a quantum circuit can apply an arithmetic operation to each basis state $\ket{z}$ (say, produce $\ket{(z+1) \mod 2^n}$), then the same operator by linearity will apply to random variables encoded in a state $\ket{\psi}_X = \sum \sqrt{f_X(x)} \ket{x}$, applying the same arithmetic operation to all outcomes of the random variable in a single application of the operator. Similar properties hold for arithmetic operations involving two random variables, such as the sum and the product. We shall introduce additional notions in Subsec.~\ref{SIsubsec:multi-qv}, to formalize such result.

\begin{example}[signed integer quantum variables]
As a second example, let us introduce the two's complement representation of a signed integer.
Let $\mathcal{D}^{\mathrm{int}} = \{ -2^{n-1}, ..., -1, 0, ..., 2^{n-1}-1 \}$. The following function
$$g^\mathrm{tc}(x) :=
\begin{cases}
    x & \text{if } x \geq 0,\\
    2^n+x & \text{otherwise}.
\end{cases}
$$
is the so-called \textit{two's complement encoding}, whose importance lies in the fact that it agrees with sum and product, modulo $n$:
$$g^\mathrm{tc}(x+y) = g^\mathrm{tc}(x)+g^\mathrm{tc}(y) \quad \text{for all } x, y \in \mathcal{D}^{\mathrm{int}} \text{ such that } x+y \in \mathcal{D}^{\mathrm{int}},$$
$$g^\mathrm{tc}(xy) = g^\mathrm{tc}(x) g^\mathrm{tc}(y) \quad \text{for all } x, y \in \mathcal{D}^{\mathrm{int}} \text{ such that } xy \in \mathcal{D}^{\mathrm{int}}.$$
If $x,y \in \mathcal{D}^{\mathrm{int}}$ but $x+y \not\in \mathcal{D}^{\mathrm{int}}$ (resp. $xy \not\in \mathcal{D}^{\mathrm{int}}$), we say that \textit{overflow} happens.
Based on $g^\mathrm{tc}$, it is trivial to define a \textit{signed integer quantum variable}.
\end{example}
For practical purposes, it can be more convenient to define the signed integer quantum variable having a parametric encoding function, $g_s: x \mapsto g^\mathrm{tc}(sx)$, where $s \in \{\pm 1 \}$ is a so-called classical sign. This way, it is possible to change the sign of a variable by altering its classical information instead of manipulating its state.

\subsection{Definition of the classically enriched quantum floating-point variable}\label{SIsubsec:qv}

Not every domain $\mathcal{D}$ and bijection $g$ are suitable to extend the quantum parallel property. For instance, our definition of floating-point registers requires a domain of uniformly spaced points. 

\begin{defn}[CEFV]
A \textit{classically enriched quantum floating-point variable (CEFV)} of length $n$ is a tuple $\mathcal{F} = (\ket{\cdot}, \mathcal{D}_{a,b}^{\mathrm{fp}}, g_{a,b}^{\mathrm{fp}}, \epsilon^-, \epsilon^+)$ where $a \in \mathbb{R}$ is the offset, $b \in \mathbb{R}\setminus\{0\}$ is the scaling factor, $\epsilon^-, \epsilon^+ \in [0, +\infty)$ are the below and above errors respectively. The domain is the uniformly spaced grid
$$\mathcal{D}_{a,b}^{\mathrm{fp}} := \left\{ a+b z \;|\; z \in \{ 0, ..., 2^n-1 \} \right\}$$
while the encoding is
$$g_{a,b}^{\mathrm{fp}}(x) := \frac{x-a}{b}.$$
\end{defn}

Allowing $b<0$ may appear redundant, and indeed it is. Nevertheless, we encompass this possibility as it is enables sign swap by simply adjusting the classical parameters, without any quantum state manipulation, similar to the classical sign $s$ we introduced above for signed integer registers.

\begin{remark}
Two quantum variables can insist on the same register $\ket{\cdot}$. This is useful in algorithms to economize on the number of qubits. For instance when one wants to add a fixed number $c \in \mathbb{R}$ to a floating-point variable, it is enough to create a new variable with a modified $a_\mathrm{out} = a_\mathrm{in} + c$, without need to instantiate a new register. It is also useful to interpret the same register both as an integer and as a floating point, by giving rise to two variables. More broadly, the quantum registers of two variables can share \textit{some} of the qubits. The paper shows in practice examples of all these circumstances.
\end{remark}
A quantum algorithm, seen as a tool operating on quantum variables, can make calculations both by (re)defining variables, and by manipulating input states through quantum circuits. Indeed, the algorithms that we introduce can be described by an output variable and a quantum circuit. Before moving to manipulation techniques for variables, let us discuss the error parameters we have defined.

In fact, the motivation for equipping CEFVs with $\epsilon^-$ and $\epsilon^+$ is to keep track of the error propagation, as we are able to compute the sum and the product only up to a rounding error. This leads to the definition of compatibility between random variables and states:
\begin{defn}[random variable $\mathcal{F}$-compatible with a state, later generalized in Definition~\ref{SIdefn:fcompatible-subreg}]\label{SIdefn:fcompatible}
Let $Y$ be a random variable valued in a domain $\mathcal{D}_Y$, not necessarily discrete, let $\mathcal{F}=(\ket{\cdot}, \mathcal{D}_{a,b}^{\mathrm{fp}}, g_{a,b}^{\mathrm{fp}}, \epsilon^-, \epsilon^+)$ be a CEFV of length $n$, and $\ket{\psi}$ an $n$-qubit state. Let $X_{\ket{\psi}}$ be the random variable $g_{a,b}^{\mathrm{fp}}$-encoded by $\ket{\psi}$. We say that $Y$ is \textit{$\mathcal{F}$-compatible with the state $\ket{\psi}$} if $- \epsilon^{-} \leq Y-X_{\ket{\psi}} \leq \epsilon^+$, see Fig.~\ref{SIfig:fcompatible}.

In this case, we write $Y \simeq_\mathcal{F} \ket{\psi}$.
\end{defn}
In other words, $Y \simeq_\mathcal{F} \ket{\psi}$ means that $\ket{\psi}$ represents an approximation of $Y$ under the encoding and the tolerances specified by $\mathcal{F}$. %
\begin{figure}
    \input{Figures/YapproxX.tex}
    \caption{Cumulative distribution function (CDF) of the target variable $Y$ (in green) and its approximating variable $X=X_{\ket{\psi}}$ (in red) in the sense of the Definition~\ref{SIdefn:fcompatible}. Notice that $X$ is bound to take values in the domain $\mathcal{D}_{a,b}^{\mathrm{fp}}$, represented by the vertical dashed grid (the first line being at $a$, and the interline space being $b$). $Y$ is $\mathcal{F}$-compatible with the state $\ket{\psi}$ as $Y$ belongs to the approximation region in light blue $[X-\epsilon^-, X+\epsilon^+]$.}\label{SIfig:fcompatible}
\end{figure}
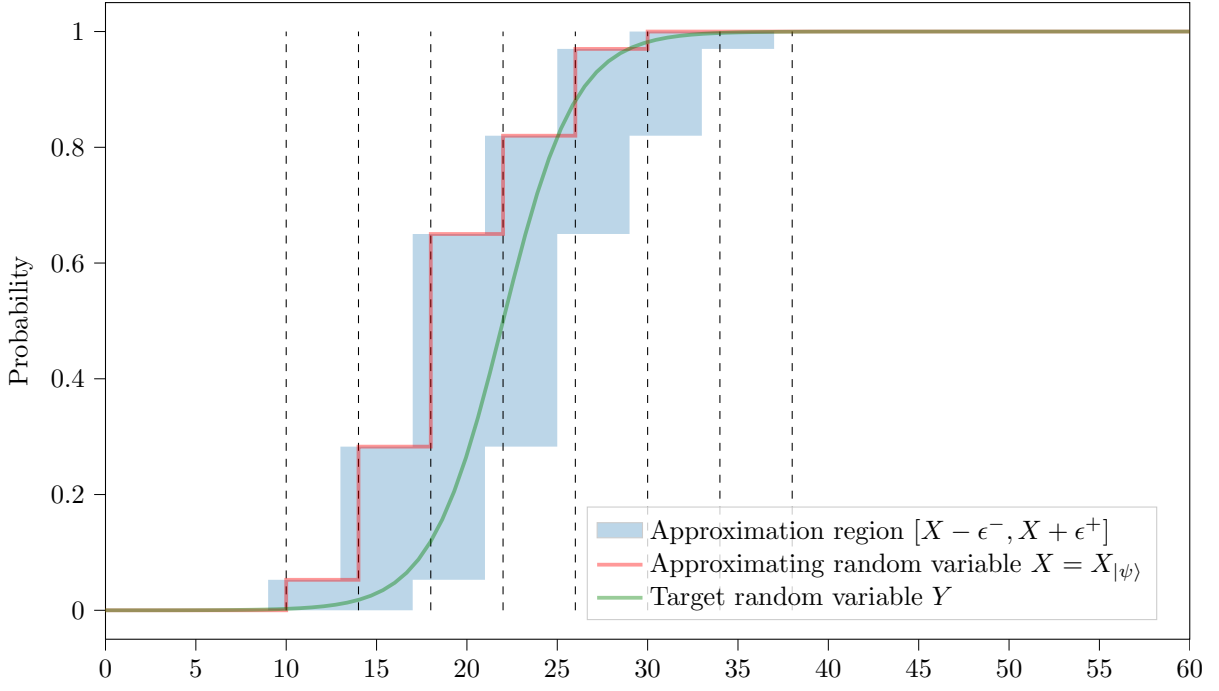

\subsection{Encoding multiple random variables in subregisters}\label{SIsubsec:multi-qv}
To make multiple quantum variables interact, it is essential to extend Definitions~\ref{SIdefn:g-encoding}, \ref{SIdefn:g-encoded} and \ref{SIdefn:fcompatible} to multiple, possibly dependent random variables in the following way:

\begin{defn}[$(g_1, ..., g_k)$-encoding of $k$ random variables]\label{SIdefn:g-encoding-mult}
Let $k$ be a positive integer. Let $n_1, ..., n_k$ be positive integers, and $\mathcal{D}_1, ..., \mathcal{D}_k$ be discrete real domains with $2^{n_k}$ points respectively. Let $g_j$ be bijections between $\mathcal{D}_j$ and $\{0, ..., 2^{n_j}-1\}$ for all $j=1,...,k$. Moreover, let $X_1, ..., X_k$ be random variables valued in $\mathcal{D}_1, ..., \mathcal{D}_k$ respectively, with a joint mass probability function $f_{X_1, ..., X_k} (x_1, ..., x_k) = \mathbb{P}(X_1=x_1, ..., X_k=x_k)$. Then we say that an $(n_1+...+n_k)$-qubit state $\ket{\psi_{(X_1, ..., X_k)}}$ is the \textit{$(g_1, ..., g_k)$-encoding for $X_1, ..., X_k$} if
$$\ket{\psi_{(X_1, ..., X_k)}} = \sum_{x_j \in \mathcal{D}_j, j=1,...,k} \sqrt{f_{X_1, ..., X_k} (x_1, ..., x_k)} \ket{g_1(x_1)} \cdots \ket{g_k(x_k)}.$$
\end{defn}

Notice that, given two random variables $X_1$ and $X_2$ and two encoding functions $g_1$ and $g_2$, one can build their separate $g_1$- and $g_2$-encodings $\ket{\psi_{X_1}}$ and $\ket{\psi_{X_2}}$, or their joint encoding $\ket{\psi_{(X_1, X_2)}}$. We have $\ket{\psi_{X_1}} \ket{\psi_{X_2}} = \ket{\psi_{(X_1, X_2)}}$ if and only if the two random variables are independent. In other terms, the encoding translates dependency of random variables into entanglement of quantum states.

Now we want to generalize the Definition~\ref{SIdefn:g-encoded} of $g$-encoded random variable, to multiple variables. The generalization is relevant in two senses. On one side, we need to infer a joint distribution of multiple random variable encoded in a single state. On the other side, we may need to discard some registers, and infer the random variables encoded only in some other registers. We collectively handle the two cases in the next Definition.

\begin{defn}[random variables $(g_1, ..., g_k)$-encoded in subregisters by a state]\label{SIdefn:g-encoded-mult}
Let $k$ be a positive integer. Let $n_1, ..., n_k$ be positive integers, and let $\mathcal{D}_1, ..., \mathcal{D}_k$ be sets of $2^{n_1}, ..., 2^{n_k}$ real numbers respectively. Let $g_j$ be a bijection between $\mathcal{D}_j$ and $\{0, ..., 2^{n_j}-1\}$ for all $j=1, ..., k$. Also, let $\ket{\cdot}_1, ..., \ket{\cdot}_k$ be registers, without qubits in common, of size $n_1, ..., n_k$ qubits, respectively. Finally, let $\ket{\psi}$ be a state on all the qubits of the registers above, and potentially on $n_a$ more qubits, which we group into an auxiliary register $\ket{\cdot}_a$. The random variables $(X_1, ..., X_k)_{\ket{\psi}}$ having joint density
$$f_{(X_1, ..., X_k)_{\ket{\psi}}}(x_1, ..., x_k) = \sum_{i=0}^{2^{n_a}-1} \abs{\braket{\psi}{g_1(x_1) \cdots g_k(x_k) i}}^2 \quad \forall (x_1, ..., x_k) \in \mathcal{D}_1 \times \cdots \times \mathcal{D}_k$$
are called the \textit{random variables $(g_1, ..., g_k)$-encoded in the subregisters $\ket{\cdot}_1, ..., \ket{\cdot}_k$ by $\ket{\psi}$}.
\end{defn}

At this point, the $\mathcal{F}$-compatibility of a random variable with a state can be extended also to the case where $\ket{\cdot}_\mathcal{F}$ of $\mathcal F$ is a subregister of the register defining the quantum system:
\begin{defn}[random variable $\mathcal{F}$-compatible with a state]\label{SIdefn:fcompatible-subreg}
Let $Y$ be a random variable valued in any domain $\mathcal{D}_Y$, let $\mathcal{F}=(\ket{\cdot}_{\mathcal{F}}, \mathcal{D}_{a,b}^{\mathrm{fp}}, g_{a,b}^{\mathrm{fp}}, \epsilon^-, \epsilon^+)$ be a CEFV of length $n$, and $\ket{\psi}$ an $m$-qubit state in a system comprising the register $\ket{\cdot}_{\mathcal{F}}$ (therefore, $m \geq n$). Let $X$ be the random variable $g_{a,b}^{\mathrm{fp}}$-encoded in the subregister  $\ket{\cdot}_{\mathcal{F}}$ by $\ket{\psi}$, in the sense of Definition~\ref{SIdefn:g-encoded-mult}. We say that $Y$ is \textit{$\mathcal{F}$-compatible (in the subregister $\ket{\cdot}_{\mathcal{F}}$) with the state $\ket{\psi}$} if $- \epsilon^{-} \leq Y-X \leq \epsilon^+$.
In this case, we write $Y \simeq_\mathcal{F} \ket{\psi}$.
\end{defn}
We refer to the former Definition only in terms of a variable \textit{$\mathcal{F}$-compatible with the state $\ket{\psi}$}, leaving implicit the fact that the compatibility is verified on the subregister $\ket{\cdot}_{\mathcal{F}}$ defined by $\mathcal F$ itself.

\section{Addition and multiplication of CEFVs}
\subsection{Prior art: evaluation of semi-boolean polynomials}\label{SIsubsec:sbp}
This section collects some results from prior literature that are useful for our theory. %

\begin{prop}[Semi-boolean polynomial (SBP) evaluation \cite{seidel_efficient_2022}]\label{SIprop:sbp}
Let $\ket{\cdot}_\mathrm{in}$ and $\ket{\cdot}_\mathrm{out}$ be two quantum registers of size $n_\mathrm{in}$ and $n_\mathrm{out}$ respectively. Let $p$ be a polynomial of $n_\mathrm{in}$ variables with integer coefficients.
Then there exists and algorithm $\SBPEval(n_\mathrm{in}, n_\mathrm{out}, p)$ that produces a quantum unitary circuit $U$ working on $\ket{\cdot}_\mathrm{in}$ and $\ket{\cdot}_\mathrm{out}$ and 1 ancillary qubit. For all $z = \sum_{j=0}^{n_\mathrm{in}-1} z_j 2^j \in \{ 0, ..., 2^{n_\mathrm{in}}-1 \}$, the circuit $U$ calculates $p$ applied to the binary digits $[z_j]$ of $z$, modulo $n_\mathrm{out}$. In formulas:
$$U \ket{z} \ket{0} \ket{0} = \ket{z} \ket{y} \ket{0},$$
where $y \equiv_{n_\mathrm{out}} p(z_0, ..., z_{n_\mathrm{in}-1})$.
In the basis made of 1- and 2-qubit gates, the gate count of $U$ is $\mathcal{O}\left( n_\mathrm{out}^2 + n_\mathrm{out} \card (m \in p) + \sum_{m \in p} c(\deg m) \right)$, uniformly in the coefficients of $p$, where $m \in p$ are the monomials in $p$, and $c(n)$ is the gate complexity of an $n$-controlled NOT gate.
\end{prop}

The number of ancillas can be increased to gain an advantage in the multiplicative constant of the gate count \cite{seidel_efficient_2022}. The function $c(n)$ can take multiple forms, but for our purposes it is enough to restrict to the hypothesis of the following Corollary. The interested reader can refer to Ref.~\cite{barenco1995elementary} for an implementation of the $n$-controlled NOT either be with $1$ ancilla qubit, $\mathcal O(2^n)$ gates and depth \cite[Lemma~7.1]{barenco1995elementary}, or with $n-1$ ancillas, and $\mathcal O(n)$ gates and depth \cite[Lemma~7.2]{barenco1995elementary}. Ref.~\cite{maslov_advantages_2016} further improved the ancillas necessary to achieve a $\mathcal O (n)$ depth, to $\left\lceil \frac{n-3}{2} \right\rceil$ for $n \geq 5$.
\begin{coroll}\label{SIcor:sbp}
Under the assumptions of Prop.~\ref{SIprop:sbp}, if furthermore $\deg p$ keeps limited when $n_\mathrm{in}$ grows, say $\deg p \leq d$, then the gate count is $\mathcal{O} \left( n_\mathrm{out}^2 + n_\mathrm{out} \cdot \card (m \in p) \right)$.
\end{coroll}
\begin{proof}
    $\deg p \leq d$ implies $\deg m \leq d$ for all monomials $m \in p$. Obviously $\sum_{m \in p} c(\deg m) \leq d \card (m \in p) = \mathcal{O}\left( n_\mathrm{out}\card (m \in p) \right)$.
\end{proof}

The former Proposition can be used for many useful arithmetic manipulations of signed and unsigned integers, as well as floating point numbers. We start be mentioning two simple ones in the following Corollary, while we refer to the original paper for other results. As said at the beginning, we build on top of Proposition~\ref{SIprop:sbp} in the results of the following subsections.

\begin{coroll}[Unsigned modular addition and multiplication \cite{seidel_efficient_2022}] \label{SIcor:sbp-sumandprod} Let $\ket{\cdot}_{1}$, $\ket{\cdot}_{2}$ and $\ket{\cdot}_\mathrm{out}$ be three quantum registers of size $n_{1}$, $n_{2}$ and $n_\mathrm{out}$ respectively. Then the \SBPEval algorithm can be used to produce unitaries $U_A$ and $U_P$ that calculate sum and product of the integer inputs into the output register, modulo $n_\mathrm{out}$, in $\mathcal{O}(n_\mathrm{out} (n_{1} + n_{2}))$ and $\mathcal{O}(n_\mathrm{out} n_{1} n_{2})$ depth respectively.
\end{coroll}
\begin{proof} Use the previous Corollary with
$$p_A = \sum_{j=0}^{n_{1}-1} z_{1, j} 2^j + \sum_{j=0}^{n_{2}-1} z_{2, j} 2^j, \qquad
p_M = \left( \sum_{j=0}^{n_{1}-1} z_{1, j} 2^j \right) \left( \sum_{j=0}^{n_{2}-1} z_{2, j} 2^j \right),$$
having $m(p_A) = n_{1} + n_{2}$ and $m(p_P) = n_{1} n_{2}$. Finally consider that WLOG we can assume $n_\mathrm{out} \leq m(p_S)$ and $n_\mathrm{out} \leq m(p_P)$.
\end{proof}

\subsection{Addition}
Algorithm~\ref{SIalgo:add} formalizes the \Add algorithm sketched in the Method section of the main manuscript.

\begin{algorithm}
\small
\Fn{\Add}{
\Input{Two CEFVs $\mathcal{F}_1$ and $\mathcal{F}_2$; a positive integer $n_\mathrm{max}$; a real number $b_\mathrm{lead} \neq 0$}
\Output{An $n_\mathrm{out}$-qubit CEFV $\mathcal{F}_\mathrm{out}$; a quantum circuit $U$}
\BlankLine
Initialize $n_\mathrm{out} \in \mathbb N$ as $n_\mathrm{out} \leftarrow n_\mathrm{max}$ \label{SIalgo:add:step1}\;
Initialize $M \in \mathbb Z$ as
\begin{algomathdisplay}
M \leftarrow  \left\lceil \log_2 \left(2^{n_\mathrm{out}} -1\right) - \log_2 \left( 
(2^{n_1}-1) \abs{\frac{b_1}{b_\mathrm{lead}}} +
(2^{n_2}-1) \abs{\frac{b_2}{b_\mathrm{lead}}} 
\right) \right\rceil \label{SIalgo:add:init-M}
\end{algomathdisplay}
For all $k \in \{ 1, 2\}$ and for all $j_k \in \{ 0, ..., n_k-1\}$, define $w_{k,{j_k}} \in \mathbb{Z}$ as the closest-integer rounding of $2^{j_k} 2^M b_k / b_\mathrm{lead}$ \label{SIalgo:add:weights} \;
Repeatedly decrease $M$ by $1$ and recalculate the weights as in the previous step, until
\begin{algomathdisplay}
\sum_{j_1=0}^{n_1-1} \abs{w_{1,j_1}} + \sum_{j_2=0}^{n_2-1} \abs{w_{2,j_2}} \leq 
2^{n_\mathrm{out}}-1
\end{algomathdisplay}
\label{SIalgo:add:iter-M}
Tighten the bound in the previous inequality by lowering $n_\mathrm{out}$ as long as the inequality holds \label{SIalgo:add:tighten-nout} \;
Define the following integer-coefficient polynomial, where the last term accounts for the sign qubit \label{SIalgo:add:poly}
\begin{algomathdisplay}
p(z_1,z_2) :=
\sum_{j_1=0}^{n_1-1} z_{1,j_1} w_{1,j_1} +
\sum_{j_2=0}^{n_2-1} z_{2,j_2} w_{2,j_2}
\end{algomathdisplay}
Apply $\SBPEval(n_1+n_2, n_\mathrm{out}, p)$
to obtain the circuit $U$ that evaluates $p$ into a register $\ket{\cdot}_\mathrm{out}$ \label{SIalgo:add:weightedsum} \;
Set the characteristics of $\mathcal{F}_\mathrm{out}$ as in Formula box~\ref{SIfb:add}
}
\caption{Addition between CEFVs with no-overflow guarantee. Refer to Prop.~\ref{SIprop:add} for the details.}\label{SIalgo:add}
\end{algorithm}

\begin{formulabox}\small
$$
\begin{aligned}
    a_\mathrm{out} &:= a_1 + a_2\\
    b_\mathrm{out} &:= 2^{-M} b_\mathrm{lead}\\
    \epsilon^\pm_\mathrm{approx} &:= 2^{-M} \abs{b_\mathrm{lead}} \sum_{k=1}^2 \sum_{j_k=0}^{n_k-1} \max \left\{\pm \left( 2^M \frac{b_k}{b_\mathrm{lead}} 2^{j_k} - w_{k,j_k} \right) \sgn b_\mathrm{lead}, 0 \right\}\\
    \epsilon^+_\mathrm{out} &:= \epsilon_1^+ + \epsilon_2^+ + \epsilon_\mathrm{approx}^+ \\
    \epsilon^-_\mathrm{out} &:= \epsilon_1^- + \epsilon_2^- + \epsilon_\mathrm{approx}^-
\end{aligned}
$$
\caption{Output variable characteristics for the addition (refer to Algorithm~\ref{SIalgo:add}).}\label{SIfb:add}
\end{formulabox}

Let us recall from the paper that, given two CEFVs $\mathcal F_1$ and $\mathcal F_2$, and a real number $b_\mathrm{lead} \neq 0$, Algorithm~\ref{SIalgo:add} produces an output CEFV $\mathcal F_\mathrm{out}$, namely, it defines $n_\mathrm{out}$, $a_\mathrm{out}$ and $b_\mathrm{out}$, as well as the error thresholds $\epsilon_\mathrm{out}^\pm$. The register $\mathcal F_\mathrm{out}$ is guaranteed to have at most $n_\mathrm{max}$ qubits. Moreover, $b_\mathrm{out} / b_\mathrm{lead}$ is a power of 2 with signed integer exponent. The algorithm also provides a quantum circuit $U$ that is able to calculate the sum of the inputs in the following sense: given two independent input states $\ket{\psi_1}$ and $\ket{\psi_2}$, living in the registers $\ket{\cdot}_1$ and $\ket{\cdot}_2$ respectively, $U$ applied to them calculates their sum up to a tolerance that is compatible with $\epsilon_\mathrm{out}^\pm$.
More in general, if $\ket{\psi_{1,2}}$ is a state representing two possibly dependent random variables in registers $\ket{\cdot}_1 \ket{\cdot}_2$, in the sense of Definition~\ref{SIdefn:g-encoding-mult}, then $U$ is still able to calculate their sum into an output register, within tolerances.
In formulas, given $Y_i \simeq_{\mathcal{F}_i} \ket{\psi_{1,2}}$ for $i=1,2$, the operator $U$ acts on such state, plus an output register which is assumed to be initialized to $\ket0$ and an auxiliary register that is also supposed to start in $\ket0$, and
\begin{equation}\label{SIeq:sum-correctness}
    U \ket{\psi_{1,2}} \ket{0} \ket{0} = \ket{\psi_{1,2, \mathrm{out}}} \ket{0} \quad \text{such that} \quad
    \left\{
    \begin{array}{lll}
    Y_i & \simeq_{\mathcal{F}_i} & \ket{\psi_{1,2,\mathrm{out}}}, \quad i=1,2,\\
    Y_1 + Y_2 & \simeq_{\mathcal{F}_\mathrm{out}} & \ket{\psi_{1,2,\mathrm{out}}},
    \end{array}
    \right.
\end{equation}
where the first condition states that inputs are unchanged, while the second condition verifies that the output register contains the sum.

The validity of Eq.~\eqref{SIeq:sum-correctness} is the correctness of the algorithm, proven in Proposition~\ref{SIprop:add}. Before, though, let us clarify the algorithm formulation by mean of some examples, collectively visualised in Fig.~\ref{SIfig:sum}.
\begin{example}\label{SIex:cefvs-sum1}
Let us start with a simple case, where $\mathcal{F}_1$, $\mathcal{F}_2$ are two variables of size $n_1 = n_2 = n$, with no offsets ($a_1=a_2 =0$) and $b_1=b_2=1$. In this case, the two registers represent two signed integers in $\{0, ..., 2^{n}-1\}$. Take $b_\mathrm{lead}=1$ as well and suppose $n_\mathrm{max} = n+1$, so that the output certainly fits into the output register without approximation. In this case, steps~\refAlgoStep{SIalgo:add:step1} to~\refAlgoStep{SIalgo:add:tighten-nout} define $n_\mathrm{out} = n+1$, $M=0$ and $w_1=w_2=[2^0, 2^1, ..., 2^{n-1}]$. The meaning of these weights, is that the least significant digit of either inputs affects only the least significant digit of the output, and so on, without any re-scaling happening. Steps~\refAlgoStep{SIalgo:add:poly} and~\refAlgoStep{SIalgo:add:weightedsum} make the addition through \SBPEval, and finally the output variable characteristics are assigned in the last step. We have $a_\mathrm{out}=0$, $b_\mathrm{out}=1$, and all $\epsilon=0$. As expected, an integer sum in suitably sized registers requires no approximation.
\end{example}

\begin{figure}
(a)
\includegraphics{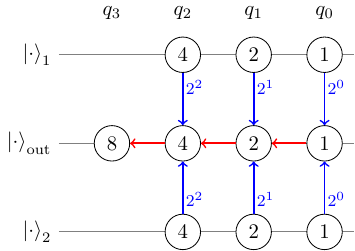}
\hfill
(b)
\includegraphics{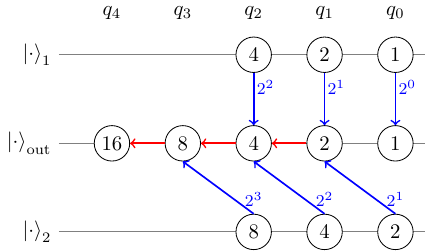}
\\[1cm]
(c) \hspace{1mm}
\includegraphics{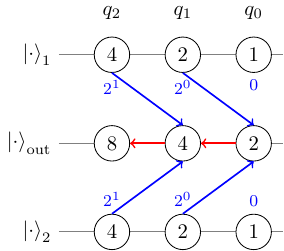}
\hfill
(d) \hspace{1mm}
\includegraphics{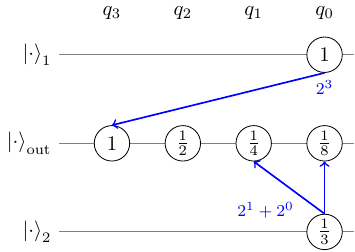}
\caption{A visual representation of how input registers (top and bottom row of each subfigure) contribute to the sum (mid row). The four cases are relative to (a) Example~\ref{SIex:cefvs-sum1} with $n=3$, (b) Example~\ref{SIex:cefvs-sum2} with $n=3$, (c) Example~\ref{SIex:cefvs-sum-approx} with $n=3$, and (d) Example~\ref{SIex:one-plus-onethird}. Each register follows the representation in Fig.~1 from the main manuscript, with the omission of $a$ which is always $0$ in the four examples. Numbers in blue are weights associated to each input qubit. Blue arrows represent the effect of such weights on the output qubits. Notice that, once a weight is decomposed as the sum of powers of two, the exponents represent the indexes of the target qubits in the output register. Red arrows mark the sum carryover accumulating along the output register.}\label{SIfig:sum}
\end{figure}

\begin{example}\label{SIex:cefvs-sum2}
Suppose $n_1 = n_2 = n$ and $a_1=a_2 =0$ as before, but $b_1=1$ and $b_2=2$ this time, with $b_\mathrm{lead}=1$. In this case, assume $n_\mathrm{max}=n+2$, so that again there is no need for re-scaling. We get $n_\mathrm{out} = n+2$, $M=0$, $w_1=[2^0, 2^1,...,2^{n-1}]$ and $w_2=[2^1, 2^2,...,2^{n}]$. In other terms, $w_2 = 2 w_1$. This way, when the sum is made, the least significant digit only contributes to the least significant digit of the output; then, the second least significant digit of the first register is coupled with the least significant digit of the second register, and so on, providing the expected result. Coherently, there is no error in the output.
\end{example}

\begin{example}\label{SIex:cefvs-sum2var}
In the previous example, if we took $b_\mathrm{lead} = 2$, this would be compensated by $M=1$, and the result would be the same. The aim of $M$ indeed is identifying the appropriate re-scale factor as to maximize the utilization of available qubits.
\end{example}

\begin{example}\label{SIex:cefvs-sum-approx}
Now take all values as in Example~\ref{SIex:cefvs-sum1}, namely $n_1=n_2=n$, $a_1=a_2=0$, $b_1=b_2=b_\mathrm{lead}=1$, except for $n_\mathrm{max} = n$ instead of $n+1$, so that an approximation needs to be made to ensure no overflow. By applying the algorithm, we get $M=1$, and $w_1=w_2=[0; 2^0,2^1,..., 2^{n-2}]$. The 0 coefficient applied to the least significant qubit of each input, means that these qubits are discarded, and everything is shifted by 1 position, to free a qubit in the result register for the carry. Consistently, a truncation arises, giving $\epsilon_\mathrm{approx}^+ = 2^{-M}\left( 2^0 2^M + 2^0 2^M \right) = 2$, that is indeed our maximal error since we discarded the least significant qubit of two registers.
\end{example}

\begin{example}\label{SIex:one-plus-onethird}
Another, more interesting and more complex example of approximation occurs when $b_1/b_\mathrm{lead}$ or $b_2/b_\mathrm{lead}$ is not a power of 2. Consider single-digit inputs $n_1=n_2=1$, $a_1=a_2=0$, $b_1=b_\mathrm{lead}=1$, but $b_2=1/3$. Take $n_\mathrm{max} = n=4$. Suppose the first register contains $\ket{z_1}=\ket{1}$, as a representation of the floating point number $1$, and the second contains $\ket{z_2}=\ket{1}$, as a representation of $1/3$. We expect the sum to be an approximation of $4/3$. Indeed we get $n_\mathrm{out}=4$, since we need as many digits as possible to accurately represent a repeating binary fraction, obviously constrained by $n_\mathrm{out} \leq n_\mathrm{max}$. We start the algorithm iterations with $M=4$, in which case $w_1=[2^4]=[16]$, $w_2=[\round (2^4/3)] = [\round (16/3)]=[5]$, where $\round$ denotes the rounding. Square brackets are kept to emphasize that in the more general case of many input qubits, $w_1$ and $w_2$ would be vectors. Such weights do not respect the exit condition of Step~\refAlgoStep{SIalgo:add:iter-M} since $16+5 > 2^4-1$. Therefore $M$ is decreased to $3$, and $w_1=[2^3]=[8]$, $w_2=[\round (2^3/3)] = [\round (8/3)]=[3]=[1+2]$. Now, $w_2 = [1+2]$ means that the (only) digit in the second register affects both the last and second-to-last result qubit, namely the third and fourth result digit respectively. The calculation outcome then writes $2^{-M} b_\mathrm{lead} (w_{1,0} z_{1,0} + w_{2,0} z_{2,0}) = 1 + \frac{1+2}{8}$, so that $1/3$ has been converted into $\frac{1+2}{8}$, that is its best representing fraction given the number of available qubits. We let the reader check that the error estimate is consistent.
\end{example}

Figure~\ref{SIfig:uniform} shows an example of repeated application of the approximate sum on 3 qubits. The sum is well performing as long as the two inputs are similarly ranged. When the cumulative sum starts dominating the input vector, approximation rises. The effect is very pronounced as a consequence of the register size, which is as little as 3 qubits. Wider registers would delay the approximation.

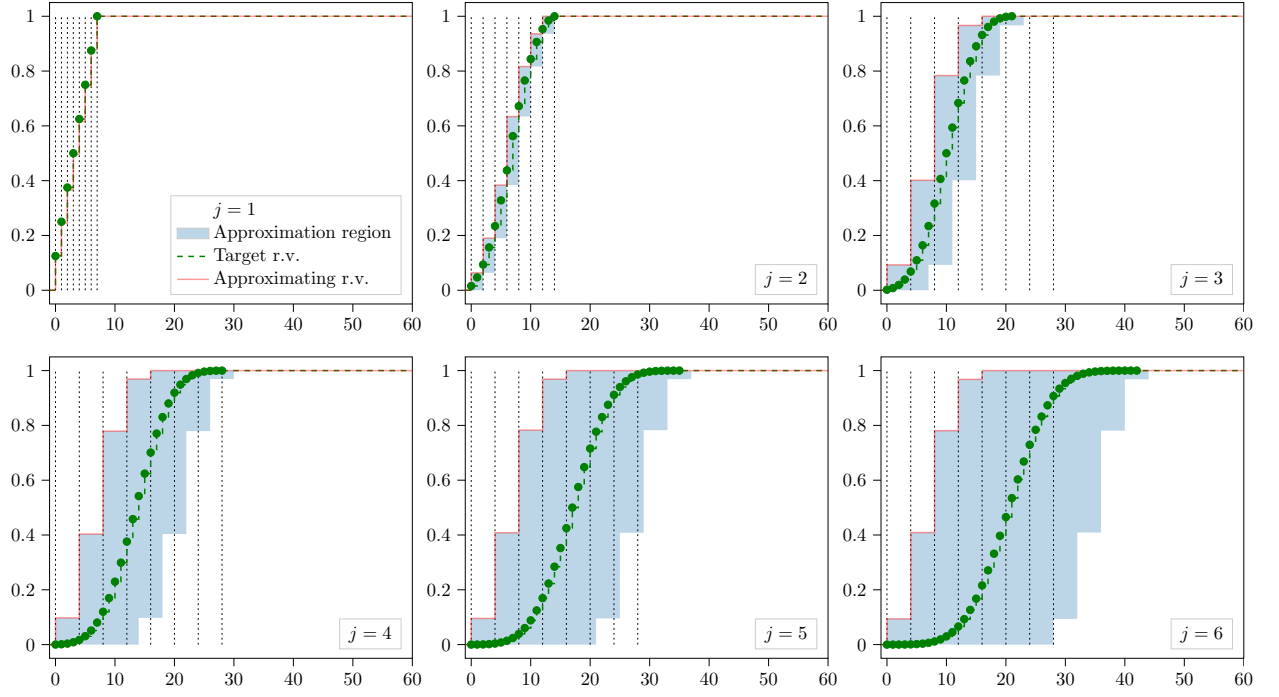
\begin{figure}
    \centering
    \input{Figures/overflow-uniform.tex}
    \caption{Repeated applications of the sum. In each plot, we represent the sum of $j$ i.i.d. random variable, for $j=1, ..., 6$. As in Figure~\ref{SIfig:fcompatible}, the representation is in form of a CDF, so that the vertical axis contains probabilities, and the horizontal axis spans the values of the random variables. The input variable is a uniform discrete distribution on 3 qubits, represented in the plot $j=1$. The sum outputs are also 3-qubit variables. The algorithm applied is Algorithm~\ref{SIalgo:add}, in its improved version according to Remark~\ref{SIrem:add-improvements}. After iteration~3, the quantum circuit remains constant, while the error grows. At that stage, indeed, since the range of the result has already become significantly wider than that of the original input, summing again a new random variable does not contribute with any significant term.}
    \label{SIfig:uniform}
\end{figure}

The previous examples lead us to the formulation of the following correctness statement:

\begin{prop}[Floating point addition]\label{SIprop:add}
Let $\mathcal{F}_1$ and $\mathcal{F}_2$ be two CEFVs of size $n_1$ and $n_2$ resp., where $\mathcal{F}_j = (\ket{\cdot}_j, \mathcal{D}_{a_j,b_j}^{\mathrm{fp}} = \mathcal{D}_j, g_{a_j,b_j}^{\mathrm{fp}} = g_j, \epsilon_j^-, \epsilon_j^+)$ for $j=1,2$. Moreover, let $n_\mathrm{max}$ be a positive integer, and $b_\mathrm{lead}$ a non-null real number. Then Algorithm~\ref{SIalgo:add} outputs an $n_\mathrm{out}$-qubit CEFV $\mathcal{F}_\mathrm{out}= (\ket{\cdot}_\mathrm{out}, \mathcal{D}_{a_\mathrm{out},b_\mathrm{out}}^{\mathrm{fp}} = \mathcal{D}_\mathrm{out}, g_{a_\mathrm{out},b_\mathrm{out}}^{\mathrm{fp}} = g_\mathrm{out}, \epsilon_\mathrm{out}^-, \epsilon_\mathrm{out}^+)$
and a quantum circuit $U$ with the following characteristics:
\begin{enumerate}
\item $n_\mathrm{out} \leq n_\mathrm{max}$
\item the width of $U$ is at most $n_1+n_2 + n_\mathrm{out}+1$
\item in the basis made of 1- and 2-qubit gates, the gate count of $U$ is
$$\mathcal{O}(n_\mathrm{out}^2 + n_\mathrm{out}(n_1 + n_2))$$

\item for $j=1,2$, given any $y_j \in \mathbb R$ and $x_j \in \mathcal{D}_j$ such that $-\epsilon_j^- \leq y_j-x_j \leq \epsilon_j^+$, $U$ behaves as follows:
$$U \ket{g_1(x_1)} \ket{g_2(x_2)} \ket{0} \ket{0} = \ket{g_1(x_1)} \ket{g_2(x_2)} \ket{z_\mathrm{out}} \ket{0}$$
where $z_\mathrm{out}$ satisfies
$$-\epsilon_\mathrm{out}^- \leq y_1+y_2-g_\mathrm{out}^{-1}(z_\mathrm{out}) \leq \epsilon_\mathrm{out}^+.$$
In other words, $x_\mathrm{out} := g_\mathrm{out}^{-1}(z_\mathrm{out})$ approximates the sum of $y_1$ and $y_2$, within the tolerances of $\mathcal{F}_\mathrm{out}$.
\end{enumerate}
\end{prop}
\begin{proof}
The first claim is trivial. For the second claim, use Prop.~\ref{SIprop:sbp}, setting $\ket{\cdot}_\mathrm{in} := \ket{\cdot}_1\ket{\cdot}_2$.
The third claim follows from Cor.~\ref{SIcor:sbp} since $\deg p = 1$, by the definition of the polynomial itself in step~\ref{SIalgo:add:poly} of Algorithm~\ref{SIalgo:add}.

Finally let us address the last claim. Write:
\begin{align*}
y_\mathrm{out}-x_\mathrm{out}
= & (y_1-x_1) + (y_2-x_2) +(x_1+x_2-x_\mathrm{out})\\
\leq & \epsilon_1^+ + \epsilon_2^+ + (x_1+x_2-x_\mathrm{out}),
\end{align*}
where we used Definition~\ref{SIdefn:fcompatible}. Now, let $z_j=g_j(x_j)$ for $j=1,2$. The choice of $n_\mathrm{out}$ made in step~\refAlgoStep{SIalgo:add:tighten-nout}, guarantees no overflow in the polynomial evaluation, so that $z_\mathrm{out} = \sum_{j_1=0}^{n_1-1} w_{1,j_1} z_{1,j_1} + \sum_{j_2=0}^{n_2-1} w_{2,j_2} z_{2,j_2}$. Then
\begin{align*}
x_1+x_2-x_\mathrm{out}
=& (a_1 + b_1 z_1) + (a_2 + b_2 z_2) - (a_1 + a_2 + b_\mathrm{out} z_\mathrm{out})\\
=& (b_1 z_1 + b_2 z_2 - 2^{-M} b_\mathrm{lead} z_\mathrm{out})\\
=& \sum_{k=1}^2 \left( b_k z_k - 2^{-M} b_\mathrm{lead} \sum_{j_k=0}^{n_k-1} w_{k,j_k} z_{k,j_k} \right)\\
=& 2^{-M} b_\mathrm{lead} \sum_{k=1}^2 \left(2^M \frac{b_k}{b_\mathrm{lead}} z_k - \sum_{j_k=0}^{n_k-1} w_{k,j_k} z_{k,j_k} \right)\\
=& 2^{-M} b_\mathrm{lead} \sum_{k=1}^2 \sum_{j_k=0}^{n_k-1} \left(2^M \frac{b_k}{b_\mathrm{lead}} 2^{j_k} - w_{k,j_k} \right) z_{k,j_k}\\
\leq& 2^{-M} \abs{b_\mathrm{lead}} \sum_{k=1}^2 \sum_{j_k=0}^{n_k-1} \max \left\{ \left( 2^M \frac{b_k}{b_\mathrm{lead}} 2^{j_k} - w_{k,j_k} \right) \sgn b_\mathrm{lead}, 0 \right\}.
\end{align*}
The proof for $\epsilon^{-}_\mathrm{out}$ is similar.
\end{proof}

\begin{remark}
The former proof holds for any $[w_{k,j_k}]_{k,j_k}$ however defined, as long as they guarantee the no-overflow condition. Our choice in step~\refAlgoStep{SIalgo:add:weights} is `optimal' in the sense that (a) $M$ is calculated in such a way to fully exploit the precision achievable in $n_\mathrm{max}$ qubits, (b) the definition of the weights minimizes and balances $\epsilon_\mathrm{approx}^+$ and $\epsilon_\mathrm{approx}^-$. Regarding (a), note that the initialization in step~\refAlgoStep{SIalgo:add:init-M} guarantees
$$\sum_{j_1=0}^{n_1-1} \abs{w_{1,j_1}} + \sum_{j_2=0}^{n_2-1} \abs{w_{2,j_2}} \geq 
2^{n_\mathrm{out}-1}-1,$$
and at the same time guarantees that step~\refAlgoStep{SIalgo:add:iter-M} requires few iterations. The proof of these facts is left to the reader.
\end{remark}

\begin{remark}\label{SIrem:blead-sum}
The choice of $b_\mathrm{lead}$ both affects the approximation degree and the circuit complexity of the sum. It is clear from the examples above, that a natural choice is picking $b_\mathrm{lead}$ equal to $b_1$ or $b_2$. In this way, the weights of the corresponding register become powers of 2, and the circuit simplifies significantly. Still, choosing $b_1$ against $b_2$ is again a non trivial decision. Since error estimates and complexity bounds can be calculated, the greedy approach of trying both alternatives is largely feasible. It should be noted thought that the choice of $b_\mathrm{lead}$ does not only impact the performance of individual sum where it is introduced, but also defines the output variable, thus conditioning the cost of all subsequent interactions of the output itself. In a highly complex algorithm, the choice of the scaling factors in intermediate registers, is an optimization task that goes beyond the purposes of the current work.
\end{remark}

\begin{remark}\label{SIrem:add-improvements}
Some improvements can be applied to the previous Algorithm.
\begin{description}
    \item[(a)] If all weights in both $w_1$ and $w_2$ share a power of 2 as a common factor, they can be divided by such factor that can be accounted for classically, thus reducing the complexity of the circuit.
    \item[(b)] After applying \textbf{(a)}, if the terms in (say) $\mathcal{F}_2$ are all multiple of $2^l$, for some positive integer $l$, and at the same time $w_{1, j_1} = 2^{j_1}$ for all $j_1=0,...,l-1$, then the $l$ least significant qubits of the sum are not impacted by $\mathcal{F}_2$, and can be copied directly from the corresponding digits in $\mathcal{F}_1$ via $\mathtt{cx}$'s. Notice that the condition on $w_{1, j_1}$ is likely verified in practice, because it is implied by $b_\mathrm{lead} = b_1$.
    \item[(c)] The previous \textbf{(b)} can be rewritten in a more general form. If only one input qubit of one input variable affects the least significant output qubit, then there is no need to make a sum on that qubit. Indeed, only one input qubit impacts the value of the output qubit, and no carryover can be present since the digit is the least significant. In such case, then, the value can be simply copied to the output qubit via a $\mathtt{cx}$. Afterwards, \textbf{(a)} can be invoked, and \textbf{(c)} again, recursively, until none of the two performs any action.
    \item[(d)] The rounding in step~\refAlgoStep{SIalgo:add:weights} can be modified to map a 0.5 decimal into 0 instead of 1. This does not affect the error, but reduces weights.
    \item[(e)] An iteration of step~\refAlgoStep{SIalgo:add:iter-M} is equivalent to rounding all weights to a power of 2. At the last iteration, a more refined technique is to round only \textit{some} weights, sufficiently many to satisfy the exit condition in step~\refAlgoStep{SIalgo:add:iter-M}, but as little as possible to limit approximation.
\end{description}
\end{remark}

\begin{remark}
For any given positive integer $L$, in the attempt to increase the result precision it is possible to calculate $L$ additional binary digits by resorting to $L$ ancilla qubits, and then uncompute these additional qubits by means of a subtraction modulo $2^L$, before returning the result. %
\end{remark}

\begin{remark}\label{SIrem:linearcomb}
Given our definition of the CEFV, one can trivially provide the appropriate \ClassicalAdd and \ClassicalProd operations, that take a CEFV and a real number in input, and output a new CEFV operating on the same register, with updated characteristics $a$ and $b$, and with no need for manipulating the quantum states. Combining \ClassicalProd with \Add, one can calculate arbitrary linear combinations \LinearCombination of two CEFVs, and specifically the difference.
\end{remark}

\begin{remark}\label{SIrem:multireg-sum}
The \Add and \LinearCombination algorithms easily extend to multiple input CEFVs, by defining $p$ appropriately in step~\refAlgoStep{SIalgo:add:weightedsum}. The multi-variable algorithm has in general smaller approximation error, compared to the repeated application of pairwise sums.
\end{remark}

\begin{remark}
The classical processing in the algorithm scales polynomially in time with the input size. This observation is vital as an exponential scaling would undermine the benefit of quantum parallel calculation.
\end{remark}

\subsection{Multiplication}

\begin{algorithm}
\small
\Fn{\Multiply}{
\Input{Two CEFVs $\mathcal{F}_1$ and $\mathcal{F}_2$; a positive integer $n_\mathrm{max}$; a real number $b_\mathrm{lead} \neq 0$}
\Output{An $n_\mathrm{out}$-qubit CEFV $\mathcal{F}_\mathrm{out}$; a quantum circuit $U$}
\BlankLine
Set $c_1 := b_1 a_2$ and $c_2 := b_2 a_1$.
Set $n_\mathrm{prod} := n_1 + n_2+1$ and $c_\mathrm{prod} := b_1 b_2$ \;
Initialize $n_\mathrm{out} \in \mathbb N$ as $n_\mathrm{out} \leftarrow n_\mathrm{max}$\;
Initialize $M \in \mathbb Z$ as
\begin{algomathdisplay}
M \leftarrow  \left\lceil \log_2 \left(2^{n_\mathrm{out}} -1\right) - \log_2 \left( 
\sum_{k=1,2, \mathrm{prod}} (2^{n_k}-1) \abs{\frac{c_k}{b_\mathrm{lead}}}
\right) \right\rceil
\end{algomathdisplay}
For all $k \in \{ 1, 2, \mathrm{prod}\}$ and for all $j_k \in \{ 0, ..., n_k-1\}$, define $w_{k,{j_k}} \in \mathbb{Z}$ as the closest-integer rounding of $2^{j_k} 2^M c_k / b_\mathrm{lead}$ \;
Repeatedly decrease $M$ by $1$ and recalculate the weights as in the previous step, until
\begin{algomathdisplay}
\sum_{j_1=0}^{n_1-1} \abs{w_{1,j_1}} + \sum_{j_2=0}^{n_2-1} \abs{w_{2,j_2}} + \sum_{j_1=0}^{n_1-1}\sum_{j_2=0}^{n_2-1} \abs{w_{\mathrm{prod},j_1+j_2}} \leq 
2^{n_\mathrm{out}}-1
\end{algomathdisplay}
Tighten the bound in the previous inequality by lowering $n_\mathrm{out}$ as long as the inequality holds \;
Define the following integer-coefficient polynomial
\begin{algomathdisplay}
\begin{aligned}
p(z_1,z_2) :=&
\sum_{j_1=0}^{n_1-1} z_{1,j} w_{1,j_1} +
\sum_{j_2=0}^{n_2-1} z_{2,j} w_{2,j_2} +
\sum_{j_1=0}^{n_1-1} \sum_{j_2=0}^{n_2-1} z_{1,j_1}z_{2,j_2} w_{\mathrm{prod},j_1+j_2}
\end{aligned}
\end{algomathdisplay}\label{SIalgo:prod:npoly} 
Apply $\SBPEval(n_1+n_2, n_\mathrm{out}, p)$
to obtain the circuit $U$ that evaluates $p$ into a register $\ket{\cdot}_\mathrm{out}$ \;
Set the characteristics of $\mathcal{F}_\mathrm{out}$ as in Formula box~\ref{SIfb:mult}
}
\caption{Multiplication between CEFVs with no-overflow guarantee. Refer to Prop.~\ref{SIprop:prod} for the details.}\label{SIalgo:prod}
\end{algorithm}

\begin{formulabox}
\small
For $k=1,2$, let:
$$
v_{k,A}^\pm := 
\begin{bmatrix}
        [\epsilon_{3-k}^+ + a_{3-k} + b_{3-k} (2^{n_{3-k}}-1)] \epsilon_k^\pm\\
        [\epsilon_{3-k}^+ + a_{3-k}] \epsilon_k^\pm \\
        [\epsilon_{3-k}^- - a_{3-k} - b_{3-k} (2^{n_{3-k}}-1)] \epsilon_k^\mp \\
        [\epsilon_{3-k}^- - a_{3-k}] \epsilon_k^\mp \\
        0
    \end{bmatrix},
    \quad
    v_{k,B}^\pm :=
    \begin{bmatrix}
    [a_{3-k} + b_{3-k} (2^{n_{3-k}}-1)] \epsilon_k^\pm\\
    [a_{3-k}] \epsilon_k^\pm \\
    [- a_{3-k} - b_{3-k} (2^{n_{3-k}}-1)] \epsilon_k^\mp \\
    [- a_{3-k}] \epsilon_k^\mp \\
    0
\end{bmatrix}
$$
Then define:
$$
\begin{aligned}
    a_\mathrm{out} :=& a_1 a_2\\
    b_\mathrm{out} :=& b_\mathrm{lead}\\
    \epsilon^\pm_\mathrm{approx} :=
    & 2^{-M} \abs{b_\mathrm{lead}} \sum_{k=1}^{2} \sum_{j_k=0}^{n_k-1} \max \left\{ \pm \left( \frac{2^{j_k} 2^M c_k}{b_\mathrm{lead}} - w_{k, j_k} \right) \sgn b_\mathrm{lead}, 0 \right\} + \\
    &+ 2^{-M} \abs{b_\mathrm{lead}} \sum_{j_1=0}^{n_1-1}\sum_{j_2=0}^{n_2-1} \max \left\{\pm \left( \frac{2^{j_1+j_2} 2^M c_{\mathrm{prod}}}{b_\mathrm{lead}} - w_{\mathrm{prod}, j_1+j_2} \right) \sgn b_\mathrm{lead}, 0 \right\}\\
    \epsilon_\mathrm{propag}^\pm :=& \min \{ \max v_{1,A}^\pm + \max v_{2,B}^\pm, \max v_{2,A}^\pm + \max v_{1,B}^\pm \} \quad \text{($\star$)}\\
    \epsilon^+_\mathrm{out} :=& \epsilon^+_\mathrm{propag} + \epsilon^+_\mathrm{approx}\\
    \epsilon^-_\mathrm{out} :=& \epsilon^-_\mathrm{propag} + \epsilon^-_\mathrm{approx}
\end{aligned}
$$
{\footnotesize {($\star$)} The notation $\max v$ represents the maximum of the (five) entries of the vector $v$.}

\caption{Output variable characteristics for the multiplication (refer to Algorithm~\ref{SIalgo:prod}).}\label{SIfb:mult}
\end{formulabox}

\begin{prop}[Floating point multiplication]\label{SIprop:prod}
Let $\mathcal{F}_1$ and $\mathcal{F}_2$ be two CEFVs of size $n_1$ and $n_2$ resp., where $\mathcal{F}_j = (\ket{\cdot}_j, \mathcal{D}_{a_j,b_j}^{\mathrm{fp}} = \mathcal{D}_j, g_{a_j,b_j}^{\mathrm{fp}} = g_j, \epsilon_j^-, \epsilon_j^+)$ for $j=1,2$. Moreover, let $n_\mathrm{max} \geq n_1, n_2$.
Then Algorithm~\ref{SIalgo:prod} outputs an $n_\mathrm{out}$-qubit CEFV $\mathcal{F}_\mathrm{out} = (\ket{\cdot}_\mathrm{out}, \mathcal{D}_{a_\mathrm{out},b_\mathrm{out}}^{\mathrm{fp}} = \mathcal{D}_\mathrm{out}, g_{a_\mathrm{out},b_\mathrm{out}}^{\mathrm{fp}} = g_\mathrm{out}, \epsilon_\mathrm{out}^-, \epsilon_\mathrm{out}^+)$
and a quantum circuit $U$ with the following characteristics:
\begin{enumerate}
\item $n_\mathrm{out} \leq n_\mathrm{max}$
\item the width of $U$ is at most $n_1 + n_2 + n_\mathrm{out} + 1$
\item in the basis made of 1- and 2-qubit gates, the gate count of $U$ is $\mathcal{O}(n_\mathrm{out} n_1 n_2)$
\item for $j=1,2$, given any $y_j \in \mathbb R$ and $x_j \in \mathcal{D}_j$ such that $-\epsilon_j^- \leq y_j-x_j \leq \epsilon_j^+$, $U$ behaves as follows:
$$U \ket{g_1(x_1)} \ket{g_2(x_2)} \ket{0} \ket{0} = \ket{g_1(x_1)} \ket{g_2(x_2)} \ket{z_\mathrm{out}} \ket{0}$$
where $z_\mathrm{out}$ satisfies
$$-\epsilon_\mathrm{out}^- \leq y_1 y_2-g_\mathrm{out}^{-1}(z_\mathrm{out}) \leq \epsilon_\mathrm{out}^+.$$
In other words, $x_\mathrm{out} := g_\mathrm{out}^{-1}(z_\mathrm{out})$ approximates the product of $y_1$ and $y_2$, within the tolerances of $\mathcal{F}_\mathrm{out}$.

\end{enumerate}
\end{prop}
\begin{proof}
The proof is similar to that of Prop.~\ref{SIprop:add}. 
The first claim is again trivial. Use Prop.~\ref{SIprop:sbp} for the second claim and Cor.~\ref{SIcor:sbp} for the third one, together with the facts that $\deg p = 2$ and the number of monomials is $\mathcal O (n_1 n_2)$.

Concerning the last claim, write:
$$
y_\mathrm{out}-x_\mathrm{out}
= (y_1 y_2 - x_1 x_2) + (x_1 x_2 - x_\mathrm{out}).
$$
The latter term is managed as in the proof of Prop.~\ref{SIprop:add}, giving rise to $\epsilon_\mathrm{approx}^\pm$.
Consider the former term, and decompose it as $(y_1 y_2 - x_1 y_2) + (x_1 y_2 - x_1 x_2)$. Let us estimate the first part:
$$
y_2 (y_1 - x_1) \leq \max \{y_2 (y_1 - x_1), 0 \} \leq
\begin{cases}
\max \{y_2, 0 \} \epsilon^+_1 &\text{where } y_1 - x_1 \geq 0, \\
\max \{- y_2 , 0 \} \epsilon^-_1 &\text{where } y_1 - x_1 < 0.
\end{cases}
$$
In turn,
$$\max \{\pm y_2, 0 \} = \max \{(\pm y_2 \mp x_2) \pm x_2, 0 \} \leq \max \{ \epsilon_2^\pm \pm a_2 \pm b_2 (2^{n_2}-1), \epsilon_2^\pm \pm a_2, 0 \},$$
so that $ y_2 (y_1 - x_1) \leq \max v_{1,A}^+ $.

Similarly, $x_1 (y_2 - x_2)$ gives $x_1 (y_2 - x_2) \leq \max v_{2,B}^+ $, and overall we derive that $y_1 y_2 - x_1 x_2 \leq \max v_{1,A}^+ + \max v_{2,B}^+$. By symmetry, decomposing $y_1 y_2 - x_1 x_2$ this time as $(y_1 y_2 - y_1 x_2) + (y_1 x_2 - x_1 x_2)$, we also obtain $y_1 y_2 - x_1 x_2 \leq \max v_{2,A}^+ + \max v_{1,B}^+$. This completes the proof for $\epsilon_\mathrm{out}^+$.

The proof for $\epsilon_\mathrm{out}^-$ flows alike.
\end{proof}

\section{Extended results}
In this Section, we provide additional details about the results presented in the paper. All experiments are run on the Qiskit \cite{qiskit} noiseless `qasm' simulator, and demonstrate the main properties of the introduced CEFVs. All simulations test Algorithm~\ref{SIalgo:add}, in its improved version according to Remark~\ref{SIrem:add-improvements}, unless otherwise stated. The \SBPEval routine was implemented following Ref.~\cite{seidel_efficient_2022}, as no implementation was readily available at the time of running the experiments. We acknowledge that the Qrisp package~\cite{qrisp} now includes \SBPEval.

\subsection{Choice of the output scaling factor}
In this Subsection we study how the choice of $b_\mathrm{lead}$ affects the algorithm, in terms of circuit depth, error and number of output qubits. We focus on the out-of-place sum described in Algorithm~\ref{SIalgo:add}, with all the improvements introduced in Remark~\ref{SIrem:add-improvements}. We start the investigation with two examples.

\begin{example}\label{SIex:b-lead1}
Like in Example~\ref{SIex:one-plus-onethird}, consider single-digit inputs, namely $n_1=n_2=1$, and a $4$-qubit output register, $n_\mathrm{out}=4$. Again, consider $a_1=a_2=0$, $b_1=1$, $b_2=1/3$. In Figure~\ref{SIfig:b-lead1}, we first observe that the depth is a step function, and partitions the $b_\mathrm{lead}$ space in areas where it is constant, as highlighted by the white and pink vertical strips. This is a consequence of the fact that weights, resulting from rounding, are constant in each band.
Looking at the right plot, in this setting, the error curve is strictly convex in each band. Consequently, there exists a unique point in the band that best represents it. This becomes clear if we specifically consider $b_\mathrm{lead}=b_2$ or $b_\mathrm{lead} = b_1$. Starting from $b_\mathrm{lead}=b_2=1/3$, marked with a black dot, no error is obtained. Indeed, the first register can contain values $\mathcal{D}_{0,1}^{\mathrm{fp}}=\{0,1\}$ and the second $\mathcal{D}_{0,1/3}^{\mathrm{fp}}=\{ 0, 1/3\}$. All these values are multiple integers of $b_\mathrm{lead}=1/3$, and so is their sum, of course. Additionally, their integer sums all fit into the output register without need for re-scaling, implying that the representation is exact.
Now consider the squared marker, representing $b_\mathrm{lead} = b_1 = 1$. In this case, the solution is heavily approximate as detailed in the Example~\ref{SIex:one-plus-onethird}.
Finally, consider $b_\mathrm{lead} = 2/3 = 2 b_2$ and $b_\mathrm{lead} = 1/2 = b_1/2$, marked with an empty circle and empty square respectively.
Since $b_\mathrm{lead} = 2/3 = 2 b_2$ can be obtained as $b_2$ times a power of 2, the $2^M$ normalization factor in the algorithm brings it back to the exact same result, and the same holds for $b_\mathrm{lead} = 1/2 = b_1/2$.
The number of qubits needed for the optimal choice $b_\mathrm{lead}=1/3$ is as low as $3$, since $1+1/3=4 b_\mathrm{lead}$, and the number $4$ requires $3$ qubits to be stored. On the contrary, all other scenarios, being approximate, exploit all the available qubits for the best approximation, namely $4$.
\end{example}

\begin{example}\label{SIex:b-lead2}
Consider now the same setting as in the previous Example~\ref{SIex:b-lead1}, keeping $b_1=1$, but setting $b_2=9/10$ instead of $1/3$.
To guarantee that the sum is exact, one should represent the common factor 90 as an integer in the available qubits, but the number of output qubits $n_\mathrm{out}=4$ does not support the representation of 90, making it impossible to perform an exact sum. In this case, the optimal $b_\mathrm{lead}$ is $b_1=1$, as shown in Figure~\ref{SIfig:b-lead2}.
\end{example}

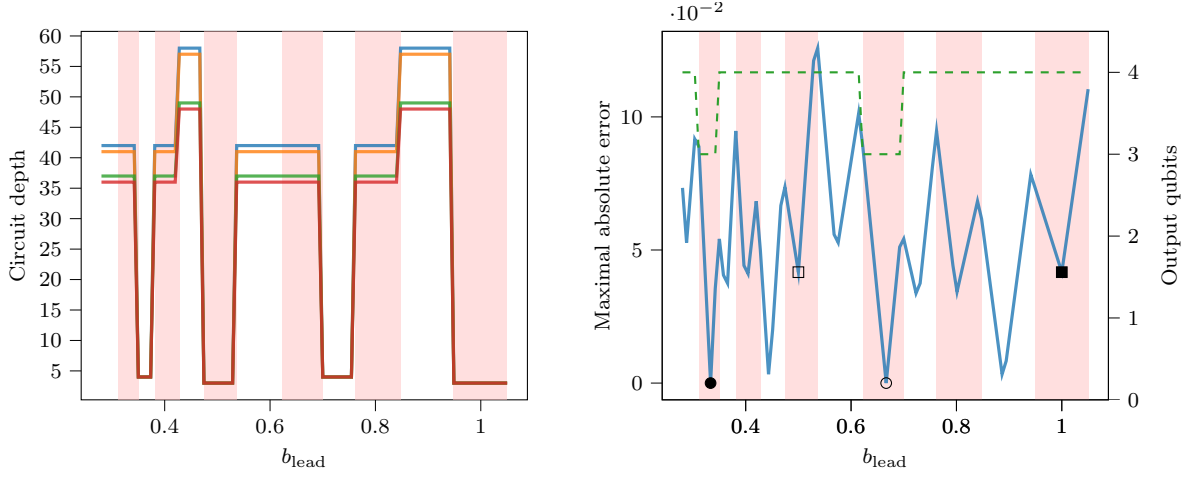
\begin{figure}[p]
\input{Figures/blead.tex}
\caption{Effect of $b_\mathrm{lead}$ on the circuit depth, the error and the number of output qubits, for the out-of-place sum described in Algorithm~\ref{SIalgo:add} and Remark~\ref{SIrem:add-improvements}.
The case is taken from Example~\ref{SIex:b-lead1}. Inputs are single-digit, namely $n_1=n_2=1$, and $a_1=a_2=0$, $b_1=1$, $b_2=1/3$. The output has $n_\mathrm{out}=4$ qubits.
In the left plot, the circuit depth (different lines correspond to different optimization levels in the qiskit transpiler, ranging from 0 to 3 from top to bottom), measured in the $[\mathtt{u3}, \mathtt{cx}]$ basis. In the right plot, the maximal approximation associated with each $b_\mathrm{lead}$, estimated as $\epsilon_{\mathrm{approx}}^++\epsilon_{\mathrm{approx}}^-$, in blue, referred to the left axis, and the number of output qubits, in dashed green, referred to the right axis.
The depth is a step function, whose changes are marked with the alternation of white and pink vertical bands in the plots.
The black dot represents the choice $b_\mathrm{lead}=b_2=1/3$, while the squared marker represents $b_\mathrm{lead} = b_1 = 1$.
The empty circle (empty square) highlights $b_\mathrm{lead} = 2/3 = 2 b_2$ (resp., $b_\mathrm{lead} = 1/2 = b_1/2$).
}\label{SIfig:b-lead1}
\end{figure}

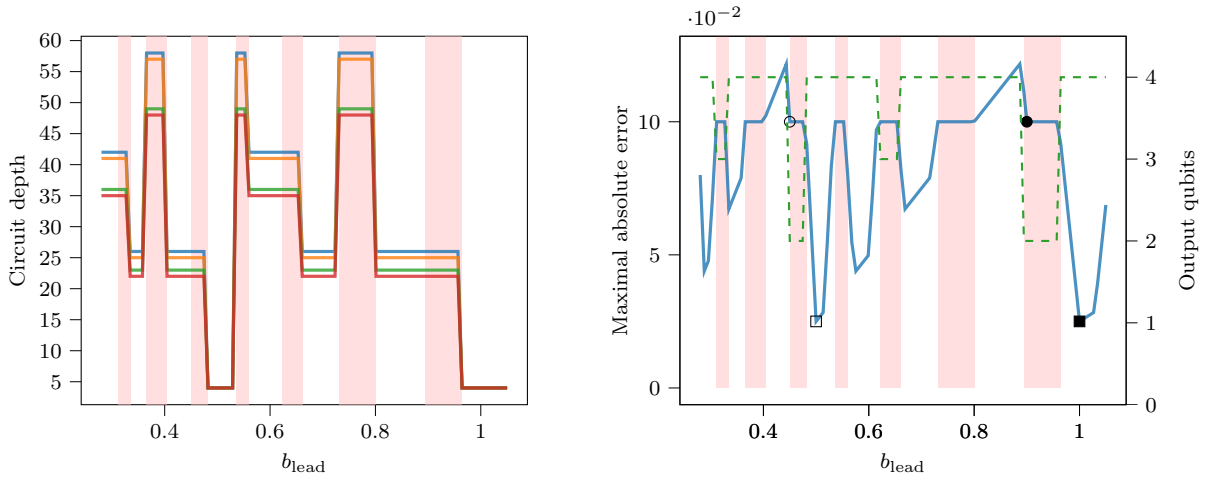
\begin{figure}
\input{Figures/blead2.tex}
\caption{Plots have the same structure as in Figure~\ref{SIfig:b-lead1}. In this case $b_1=1$ and $b_2=9/10$, thus representing Example~\ref{SIex:b-lead2}.
}\label{SIfig:b-lead2}
\end{figure}

\begin{figure}
    \input{Figures/blead-summary.tex}
\caption{Effect of the choice of $b_\mathrm{lead}$ on errors, for the sum of two CEFVs of size $n$, in an output register of same size $n$. (a) $n=8$. In the top plot, $b_1$ is kept equal to 1, $b_2$ varies from $10^{-3}$ to $10^4$, and $b_\mathrm{lead}$ in $[1,2]$. %
The color intensity represents the relative error, namely $\frac{\epsilon^++\epsilon^-}{b_1 2^n + b_2 2^n}$. In the bottom plot, a focus on three specific choices of $b_\mathrm{lead}$, namely $b_1$ (i.e. 1), $b_2$, and the optimal $b_\mathrm{opt}$. %
(b) Zoom on the range $b_2 \in [1,2]$, this time in linear scale (the corresponding area in (a) is marked with a dotted rectangle).  (c) Same as (a), but $n=4$.} \label{SIfig:b-lead-summary}
\end{figure}
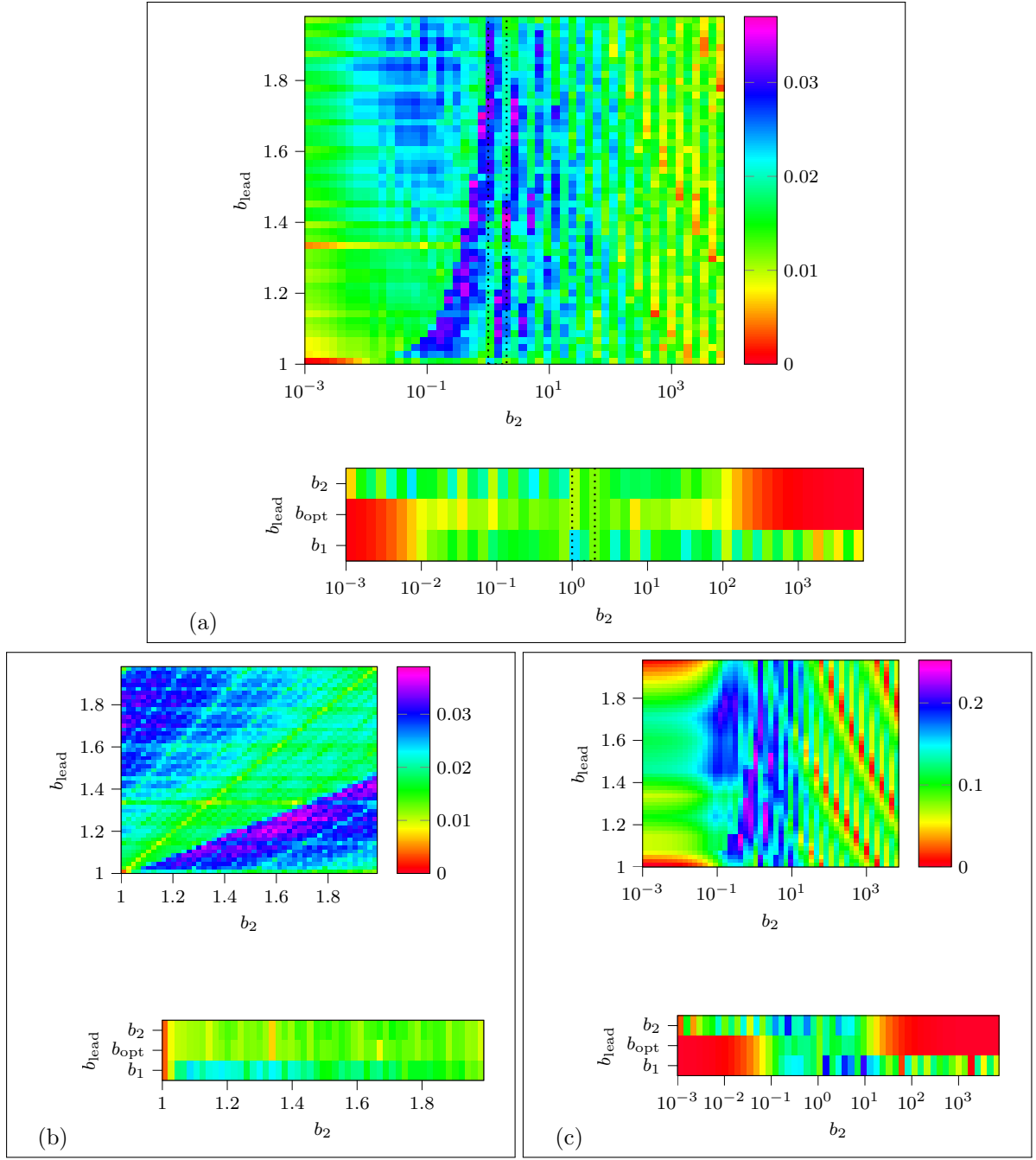

The two examples above show that the choice of the output $b_\mathrm{lead}$ has a direct effect on the approximation level, and sometimes also on the circuit depth and qubit count. Remarkably, in Example~\ref{SIex:b-lead1}, $b_2$ is the optimal choice, while in the Example~\ref{SIex:b-lead2}, $b_1$ is.

For a more comprehensive view, a numerical study of the relative error $\frac{\epsilon_{\mathrm{approx}}^++\epsilon_{\mathrm{approx}}^-}{b_1 2^n + b_2 2^n}$ was introduced in the main manuscript:
Fig.~2 therein, is here extended into Figs.~\ref{SIfig:b-lead-summary}{(a)}, {(b)} and {(c)}.

As already observed in the paper body, the best between $b_1$ and $b_2$ is a good choice for $b_\mathrm{lead}$, making in-place calculation a viable option, and confirming our observation in Examples~\ref{SIex:b-lead1} and~\ref{SIex:b-lead2}. Moreover, such fact justifies the applicability of the improvements (b) and (c) in Remark~\ref{SIrem:add-improvements}. We demonstrate the dramatic impact of Remark~\ref{SIrem:add-improvements} in Subsection~\ref{SIsubsec:add-improvements}.

\subsection{Advantages of the offset}
Figure~\ref{SIfig:step} is an extension of Fig.~3 in the main manuscript, highlighting the improvement introduced by the offset, in terms of precision achieved in a given number of qubits, as well as the limitation of error propagation after repeated arithmetic operations.

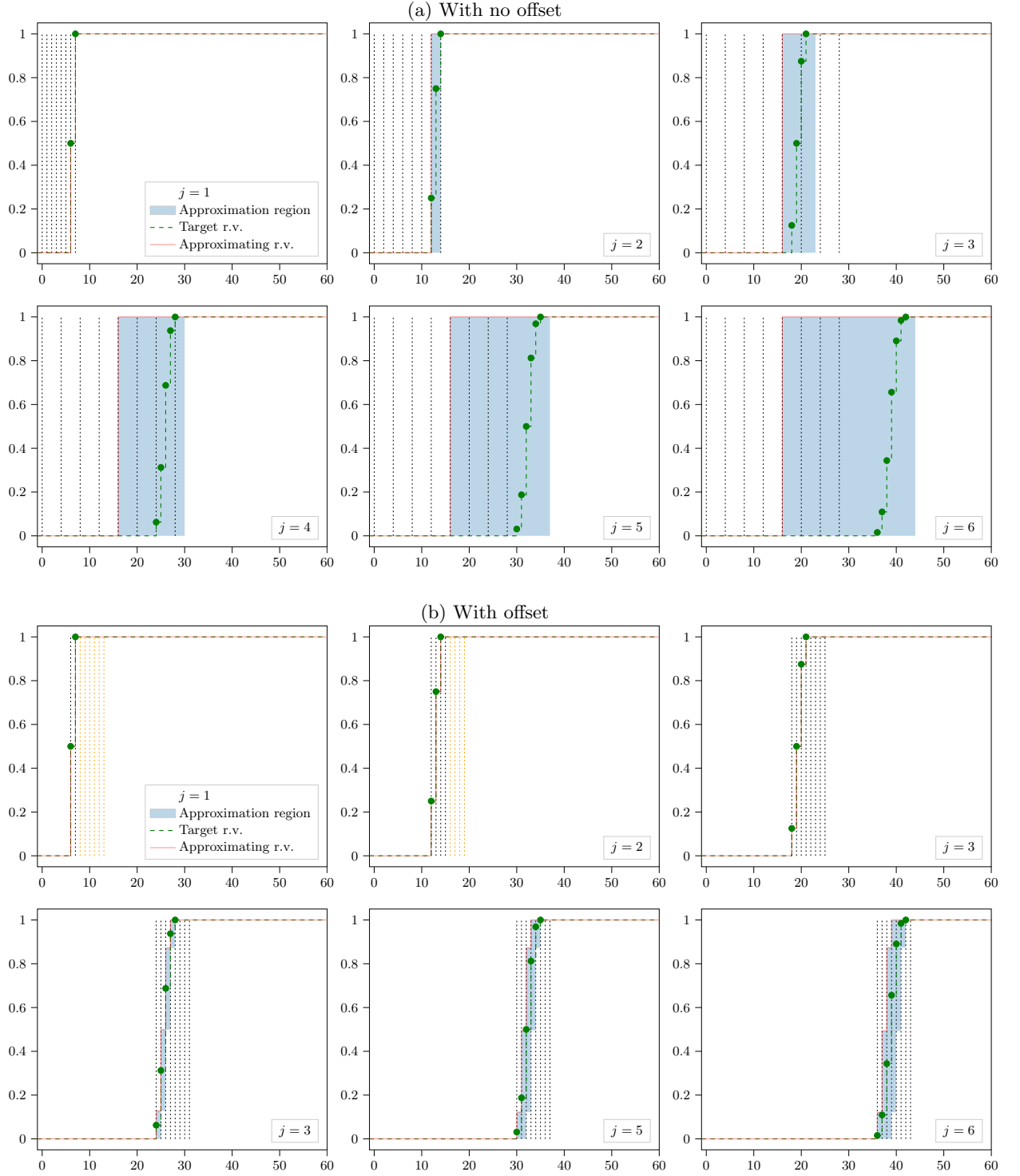
\begin{figure}
    \small
    \centering
    {(a) With no offset}\\
    \input{Figures/overflow-step-nooffset.tex}
    \\[4mm]
    {(b) With offset}\\
    \input{Figures/overflow-step-offset.tex}
    \caption{CDF for repeated applications of the sum. The plot representation is as in Figure~\ref{SIfig:fcompatible}. Here the same input random variable (providing 6 or 7 both with 50\% probability) is represented (a) with no offset, in 3 qubits, and (b) with offset, in 1 qubit. The result variables are constrained to 3 qubits at most. The different effect on error propagation is clearly visible. In the case with offsets, the vertical dashed grid of representable numbers is denser in the relevant region, implying lower approximation.}
    \label{SIfig:step}
\end{figure}

\subsection{Effects of the algorithmic improvements and in-place implementation}\label{SIsubsec:add-improvements}
Remark~\ref{SIrem:add-improvements} contains some improvements to Algorithm~\ref{SIalgo:add}. Despite being simple to explain, such modifications provide a significant boost to the performance, particularly in terms of circuit depth. 

Figure~\ref{SIfig:improvements} demonstrates the behavior of depth and error, where $b_1=1$ is fixed, and $b_2 \in \{1.4, \allowbreak 3, \allowbreak 0.25,\allowbreak 1/(2^5-1)\}$ is varied in the different test cases. Four scenarios are considered, rising from the combination of two parameters, as the sum can be performed on an external register rather than in place, and the improvements in Remark~\ref{SIrem:add-improvements} can be active or inactive.

The last pathological example shows an improvement of $99\%$ in depth, from $187$ to $2$, for $n_\mathrm{out}=9$ when not run in-place. When performing addition in-place, depth is usually better, and error is comparable. This observation, complemented with the fact that in-place operations also have reduced width, as they do not require an additional output register, determines that in-place is generally the preferable choice.

\begin{figure}
    \centering
    \include{Figures/improvements.tex}
    \vspace{-1cm}
    \caption{Effect of algorithmic improvements suggested in Remark~\ref{SIrem:add-improvements}. Each plot tests four cases, given by the combinations of two parameters: result on external register (`ext') vs in-place (`inpl'), and improvements in Remark~\ref{SIrem:add-improvements} off (`non-impr') vs on (`improved'). When the in-place case is not plotted, it means that overflow is happening. The underlying problem is the sum of two registers with $n_1=n_2=4$, $b_1=b_\mathrm{lead}=1$ (to make the `inpl' case comparable with `ext'), $n_\mathrm{out}$ varying on the abscissa, and (a) $b_2=1.4$, (b) $b_2=3$, (c) $b_2=0.25$, (d) $b_2 = 1/(2^5-1)$. Circuit depths are measured in the basis $[\mathtt{u3}, \mathtt{cx}]$ with the optimization level 3 of the qiskit transpiler. Algorithmic improvements reduce significantly the circuit depth, and sometimes also the error as in (a). In-place addition is generally more convenient than the counterpart, but (c) shows an exception. The case (d) is a pathological example demonstrating dramatic improvement on the depth.}
    \label{SIfig:improvements}
\end{figure}

\clearpage
\bibliography{main}{}
\bibliographystyle{ieeetr}

%% file: Figures/YapproxX.tex
\includegraphics{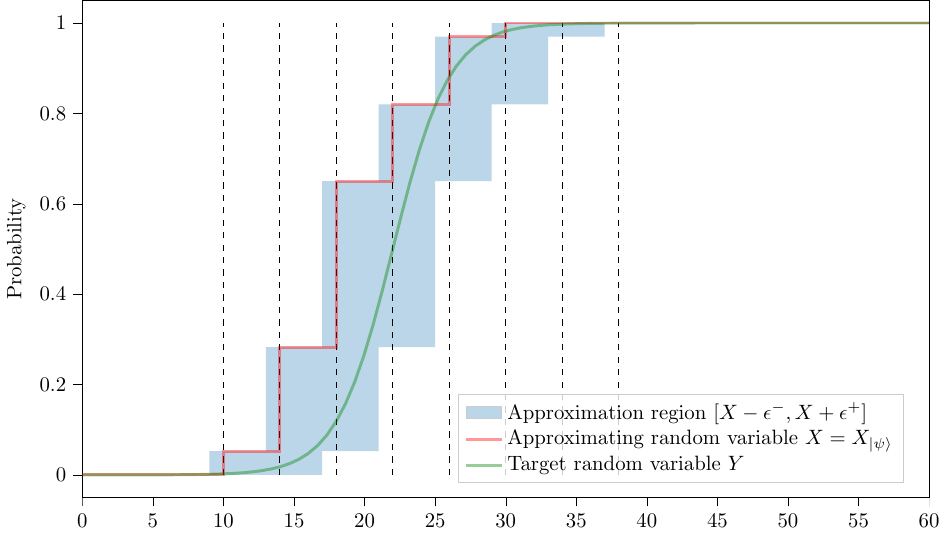}

%% file: Figures/overflow-uniform.tex
\includegraphics{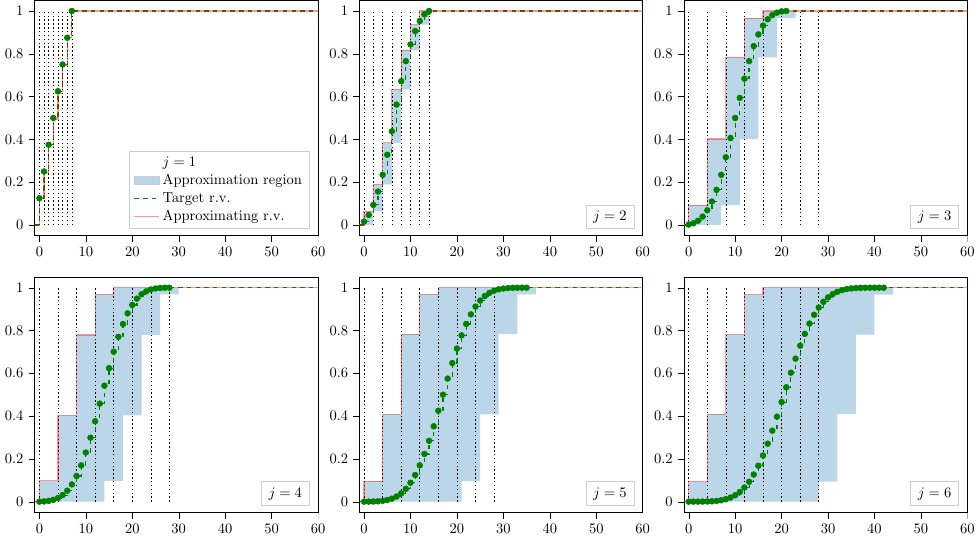}

%% file: Figures/blead.tex
\includegraphics{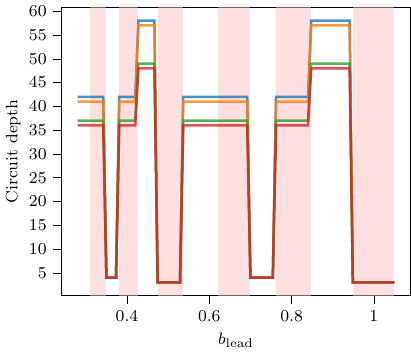}
\hfill
\includegraphics{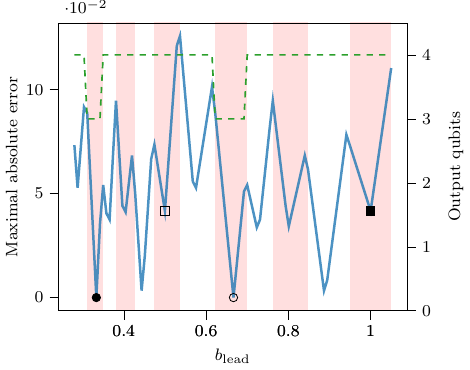}

%% file: Figures/blead2.tex
\includegraphics{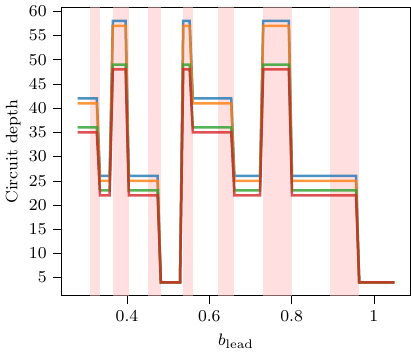}
\hfill
\includegraphics{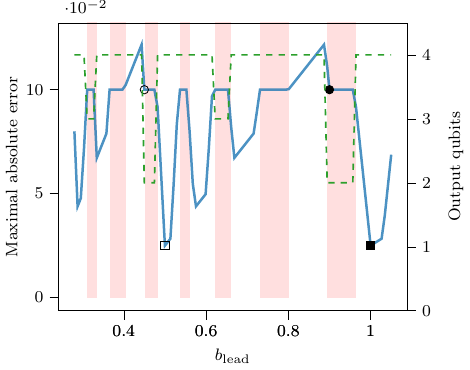}

%% file: Figures/blead-summary.tex
\centering
\fbox{
\begin{minipage}{.7\textwidth}
\vspace{1mm}
\centering
\includegraphics{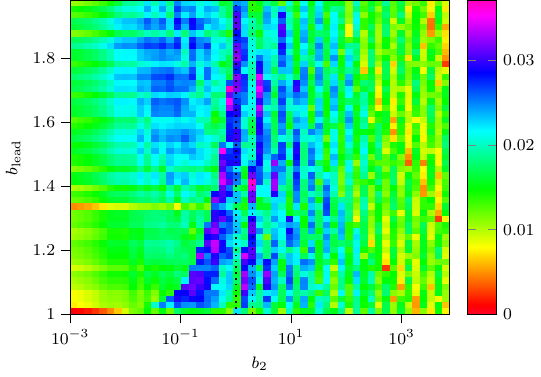}
\\[5mm]
(a) \hspace{4mm}
\includegraphics{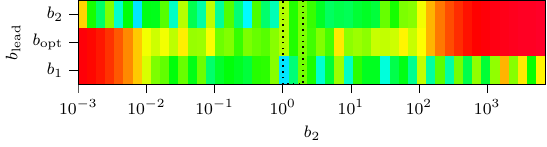}
\end{minipage}
}
\\[1mm]
\fbox{
\begin{minipage}[c][75mm]{.46\textwidth}
\centering
\vspace{1mm}
\includegraphics{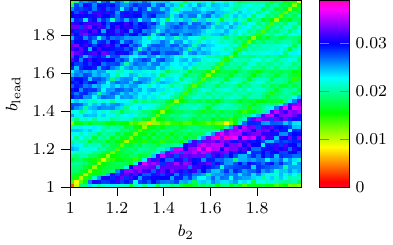}
\vfill
(b)
\includegraphics{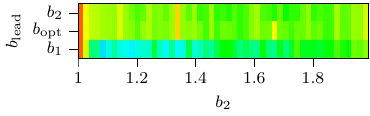}
\end{minipage}
}
\hfill
\fbox{
\begin{minipage}[c][75mm]{.46\textwidth}
\centering
\includegraphics{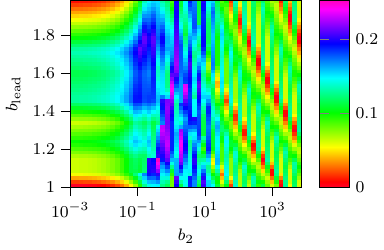}
\vfill
(c)
\includegraphics{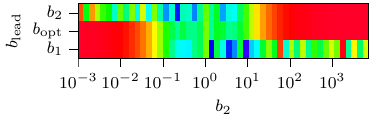}
\end{minipage}
}

%% file: Figures/overflow-step-nooffset.tex
\includegraphics{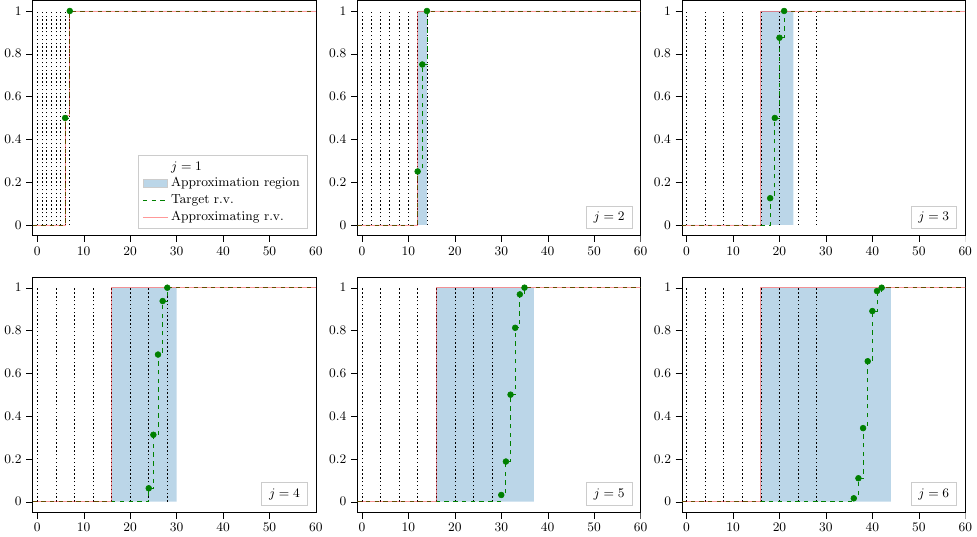}

%% file: Figures/overflow-step-offset.tex
\includegraphics{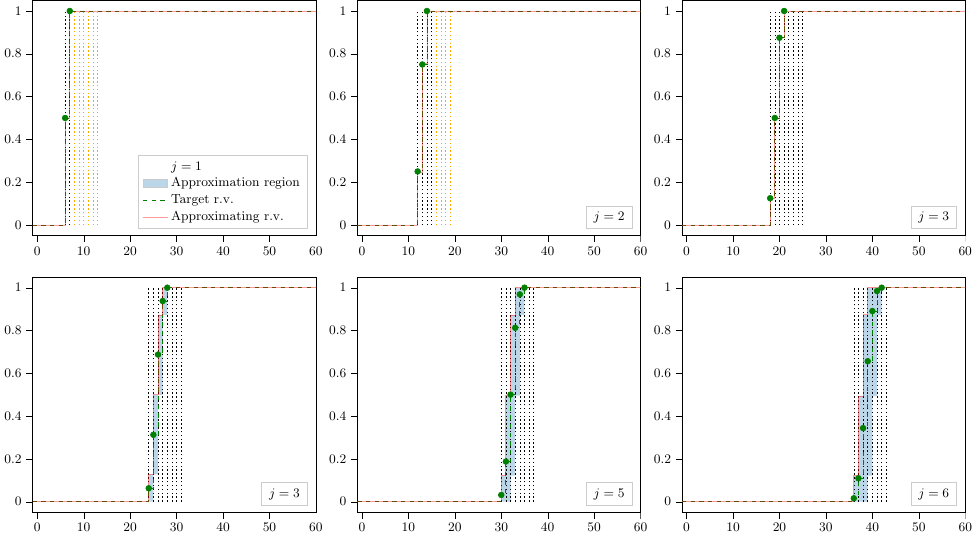}

%% file: Figures/improvements.tex
\raisebox{5mm}{{(a)}}
\includegraphics{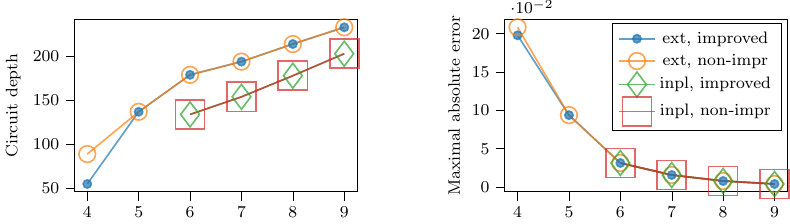}

\raisebox{5mm}{{(b)}}
\includegraphics{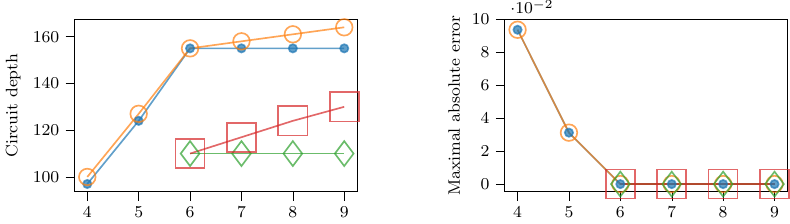}

\raisebox{5mm}{{(c)}}
\includegraphics{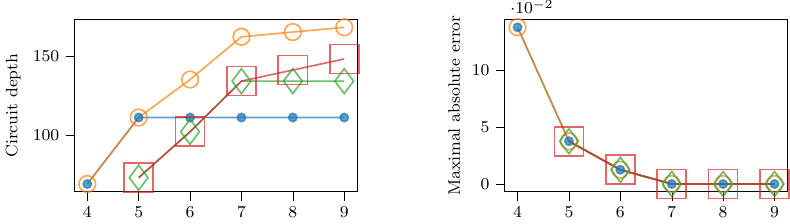}

\raisebox{5mm}{{(d)}}
\includegraphics{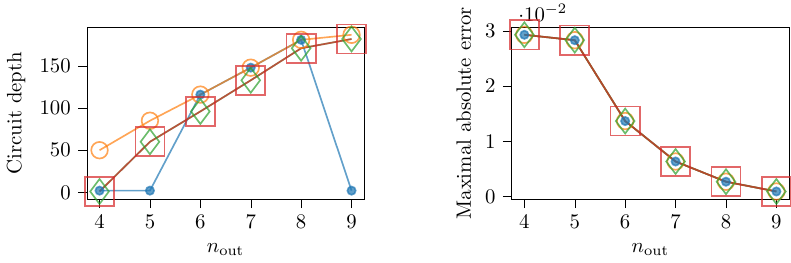}